\documentclass[11pt]{article}
\usepackage{amsmath,amsfonts,amssymb,amsthm,  mathrsfs,mathtools}
\usepackage[margin=1in]{geometry}
\usepackage{hyperref}
\hypersetup{
  colorlinks,
  citecolor=violet,
  linkcolor=blue,
  urlcolor=blue}
\usepackage{dsfont}
\usepackage{algorithm}
\usepackage{algorithmic}
\usepackage{subcaption}
\usepackage{xfrac}
\usepackage{enumitem}
\usepackage{graphicx}
\usepackage{xcolor}
\usepackage[square,numbers]{natbib}
\usepackage{algorithm}

\usepackage[utf8]{inputenc} 
\usepackage[T1]{fontenc}    
\usepackage{hyperref}       
\usepackage{url}            
\usepackage{booktabs}       
\usepackage{amsfonts}       
\usepackage{nicefrac}       
\usepackage{microtype}      
\usepackage{xcolor}         

\usepackage{microtype}
\usepackage{graphicx}
\usepackage{booktabs} 

\usepackage{amsmath}
\usepackage{amssymb}
\usepackage{mathtools}
\usepackage{amsthm}
\usepackage{makecell}
\usepackage{dsfont}
\usepackage{pgfplots}
\usepackage{multirow}
\usepackage{rotating}
\usepackage{enumitem}
\usepackage{esvect}
\usepackage{multirow}
\usepackage{tikz}
\usetikzlibrary{patterns}
\usetikzlibrary{arrows}

\usepackage{times}
\usepackage{url}
\usepackage{setspace}
\usepackage{amsfonts}       
\usepackage{nicefrac}       
\usepackage{microtype}      
\usepackage{mathtools}
\usepackage{bbm}
\usepackage{float}
\usepackage{subcaption}
\usepackage{caption}
\usepackage{algorithmic}
\usepackage{algorithm}

\usepackage{authblk}
\usepackage{booktabs}
\usepackage{multirow}

\newcommand{\syn}{\textbf{\textsc{Syn}}}
\newcommand{\news}{\textbf{\textsc{News}}}
\newcommand{\cifar}{\textbf{\textsc{CIFAR}}}

\newcommand{\rg}{\textbf{\textsc{Greedy}}}
\newcommand{\cg}{\textbf{\textsc{ConGreedy}}}

\newcommand{\algtwo}{\textbf{\textsc{LPRMono}}}
\newcommand{\algone}{\textbf{\textsc{LPR}}}

\renewcommand{\algorithmiccomment}[1]{\quad\quad // #1}

\theoremstyle{plain}
\usepackage{wrapfig}
\newtheorem{theorem}{Theorem}[section]
\newtheorem{proposition}[theorem]{Proposition}
\newtheorem{lemma}[theorem]{Lemma}
\newtheorem{fact}[theorem]{Fact}

\theoremstyle{definition}
\newtheorem{definition}[theorem]{Definition}

\theoremstyle{remark}

\newtheorem{example}{Example}
\newtheorem{claim}{Claim}

\usepackage[textsize=tiny]{todonotes}

\usepackage[capitalize,noabbrev]{cleveref}

\DeclareMathOperator*{\argmax}{argmax}

\newcommand{\cov}{\mathcal{C}}

\newcommand{\bcov}{\mathnormal{\bar{\cov}}}

\newcommand{\bX}{\mathnormal{\bar{X}}}
\newcommand{\bC}{\mathnormal{\bar{C}}}

\newcommand{\pval}{\mathnormal{\rho}}

\newcommand{\bidset}{\mathnormal{\Lambda}}

\newcommand{\bidprof}{\mathnormal{\boldsymbol{b}}}

\newcommand{\priceprof}{\mathnormal{\boldsymbol{p}}}

\newcommand{\bundle}{\mathcal{D}}
\newcommand{\rad}{\mathnormal{r}}
\newcommand{\welf}{\mathcal{W}}

\newcommand{\db}{\mathcal{DB}}

\DeclareMathOperator{\E}{\mathbb{E}}
\DeclareMathOperator{\Prob}{\mathbb{P}}

\newcommand{\norm}[1]{\left\lVert#1\right\rVert}

\newcommand{\optlp}{\text{OPT}_{\text{LP}}}

\begin{document}

\title{Data Auctions for Retrieval Augmented Generation}
\date{}  

\newcommand*{\affaddr}[1]{#1}
\newcommand*{\affmark}[1][*]{\textsuperscript{\normalfont#1}}
\newcommand*\samethanks[1][\value{footnote}]{\footnotemark[#1]}

\author{
  Minbiao Han\affmark[1]\thanks{These authors contributed equally to this work.}~,  Seyed A. Esmaeili\affmark[1]\samethanks~, Michael Albert\affmark[2]~,  Haifeng Xu\affmark[1]\\
  \affaddr{\affmark[1] University of Chicago}\\
  \affaddr{\affmark[2] University of Virginia}\\
}

\date{}  
\maketitle

\begin{abstract}
  We study the problem of data selling for Retrieval Augmented Generation (RAG) tasks in Generative AI applications. We model each buyer's valuation of a dataset with a natural coverage-based valuation function that increases with the inclusion of more relevant data points that would enhance responses to anticipated queries. Motivated by issues such as data control and prior-free revenue maximization, we focus on the scenario where each data point can be allocated to only one buyer. We show that the problem of welfare maximization in this setting is NP-hard even with two bidders, but  design a polynomial-time $(1-1/e)$ approximation algorithm for any number of bidders. Unfortunately, however, this efficient allocation algorithm fails to be incentive compatible. The crux of our approach is   a carefully tailored  post-processing step called \emph{data burning} which retains the $(1-1/e)$ approximation factor but achieves incentive compatibility. Our thorough experiments on synthetic and real-world image and text datasets  demonstrate the practical effectiveness of our algorithm compared to popular baseline algorithms for combinatorial auctions. 
\end{abstract}

\section{Introduction}
Generative AI (GenAI) 
has shown impressive performance in various domains such as text,  code, image, and even video generation \cite{ho2020denoising, ho2022video,touvron2023llama,achiam2023gpt,jiang2024survey}. The impressive performance of GenAI has had profound impacts on numerous applications. Perhaps the most important and salient application of GenAI is its usage as a replacement for search engines and question-answer websites. In particular, it has become widespread to look up information using   GenAI platforms such as ChatGPT, Perplexity, and Claude;  traditional websites such as the Google and Bing have also developed similar   ``AI Overview'' functionality.  The answers provided by GenAI models are rich, detailed, and summarize information across various webpages and sources. 

Although GenAI models have shown promise as a tool for acquiring knowledge, they are known to \emph{hallucinate}, i.e., output invalid or even fully nonfactual answers \cite{huang2023survey}. 
Retrieval Augmented Generation (RAG) \cite{lewis2020retrieval,asai2023retrieval,gao2023retrieval}  has emerged as a very useful method that not only gives the user 
confidence in the GenAI answer but also enables the user to look up the exact source of a specific piece of information in the answer. In particular, by using RAG the output of the GenAI includes \emph{references} to the sources of the information.

Another important usage of RAG is to incorporate information into a response that was not present during training.
For example, common GenAI chatbots can converse intelligently about current events even though the events occurred after the training data cutoff date.
This is due to news articles being incorporated into the context window through RAG.
Many news organizations (The Associated Press, News Corp, and The Financial Times to give a few examples) have entered into licensing agreements with GenAI providers to allow for their news articles to be used in this way.

In this paper, we consider a setting where a collection of GenAI platforms are interested in buying information from a seller who possesses a large \emph{database} of information for RAG based usage. These GenAI platforms  anticipate  queries about certain topics, are interested in augmenting their generated answers with the information from the database coupled with references. The seller looks to auctions off the database and can allocate a dataset from the database (i.e., a subset of data points)  to each buyer and assign a price for each dataset. We focus on the more challenging setting where the seller does \emph{exclusive} allocation (i.e.,   each data point can be allocated to at most one buyer).  
Such exclusive data selling could arise from multiple reasons. 
First, for sensitive data, selling each data point exclusively will reduce the likelihood that the data is leaked by the buyer, since the seller can easily trace the source of the leakage to the buyer.
Second, in many cases buyers would want exclusive allocations  since it  gives them a competitive advantage, leading to higher valuations for the dataset and increased revenue for the seller as noted in \cite{agarwal2024towards}. 

Motivated by the above questions, we solve the mechanism design problem of allocating datasets for RAG based applications with private buyer preferences. Our contributions are summarized as follows. Due to space constraints, all proofs of theoretical results are provided in the appendix.
\begin{enumerate}[leftmargin=*]
    \item \textbf{Model and Valuation Function:}  Section \ref{sec:model} introduces the formal model of our problem including the mechanism design objectives and a new RAG-based coverage function that naturally captures a buyer's value for any subset of data points and turns out to satisfy submodularity. Interestingly, we further show that in our setting, revenue maximization can be reduced to welfare maximization.   
  
    \item \textbf{Computational Hardness and Efficient Approximation:}  Section \ref{sec:new_algorithmic}  analyzes the computational complexity of our problem. We show that the demand oracle  (see Definition \ref{def:demand_oracle}), which is commonly used in combinatorial auctions, is NP-hard to implement in our setting. Moreover,  welfare maximization is NP-hard even for special cases such as two bidders (though trivial for one bidder). On the positive side,   we introduce a  $(1-1/e)$-approximate algorithm for general welfare maximization. Unlike many previous submodular welfare maximization algorithms, our algorithm is practical, and involves a simple linear program and a rounding step. 
    
    \item \textbf{Incentive Compatible Welfare Maximizing Mechanism:}  Section \ref{sec:new_ic_theory} accounts for the issue of incentives. We start by showing an interesting generalization of the well-known Myerson lemma \cite{myerson1981optimal}  to our setting with combinatorial allocations. In particular, given any allocation satisfying a certain monotonic property, we can find a payment scheme that induces an incentive compatible mechanism, i.e., truthful bidding is optimal in expectation for each bidder.   
    However, monotonicity turns out to be a quite delicate property -- we find that almost all natural submodular maximization algorithms turn out to \emph{not} be monotonic, and neither is our approximate algorithm from Section \ref{sec:new_algorithmic}. However, we develop a novel post-processing technique for the output of our previous allocation algorithm that leads to monotonic allocation and meanwhile preserves the  $1-1/e$ approximation. This helps us to achieve incentive compatibility at no additional performance cost to our original approximation algorithm. We also discuss potential further applications of our techniques to related problems.  
    \item \textbf{Experiments:}  Section \ref{sec:experiment} tests our algorithm  on real-world datasets and demonstrates its superior performance over popular baselines in terms of both revenue quality and incentive compatibility.   
\end{enumerate}

\subsection{Additional Discussion on Related Work}
Our work is closely related to  research on AI and computational economics, as detailed below.

\textbf{Mechanism design for generative AI.} Conceptually, our work contributes to the recent literature on mechanism design for generative AI. While a significant body of work has explored the use of auctions to shape  LLM outputs and integrate ads \cite{duetting2024mechanism,dubey2024auctions,hajiaghayi2024ad,soumalias2024truthful,sun2024mechanism}, our approach takes a different direction. We investigate methods for selling data to accomplish generative AI tasks using LLMs, offering a novel perspective on the intersection of economics and AI.

\textbf{Data markets.} Our work also intersects with the field of data market design, which encompasses data valuation and data selling. Previous work has introduced various ML markets \cite{agarwal2019marketplace,chen2019towards,liu2021dealer,han2023data}, primarily focusing on selling the instrumental value of data by creating ML models from aggregated market data and offering different versions to buyers based on their preferences. While \citet{mehta2021sell} explored a similar data selling problem as our paper, their focus was on pricing conventional data types, such as phone lists. Additionally, their model assumes that one buyer's purchase decision does not impact the utilities or decisions of other buyers, which differs from our exclusive dataset selling approach. In a related but distinct area, there's a substantial body of research on data valuation using the concept of Data Shapley \cite{ghorbani2019data,jia2019towards,tang2021data,schoch2022cs,wang2024efficient}. This method assesses the value of individual data points or groups by measuring their impact on model performance. Our approach diverges from this by modeling data value independently of the specific model used.

\textbf{Combinatorial auctions.} Our data auction problem also shares notable similarities with traditional combinatorial auction problems. A notable related work by \citet{dobzinski2005approximation} presents an incentive compatible mechanism with a $1/\sqrt{n}$ approximation algorithm for general submodular value functions. \citet{dughmi2011convex}   proposed a mechanism that is truthful-in-expectation and achieves a $(1 - 1/e)$-approximation. However, their algorithm only attains an expected polynomial runtime and requires a stronger ``lottery-value oracle.'' This oracle can be approximated using a traditional value oracle but would need infinitely many value queries to achieve an exact truthful mechanism due to sampling errors. Other relevant studies in this area have limitations such as focusing on data auctions with single-minded buyers \cite{lehmann2002truth,dobzinski2007mechanisms,mu2008truthful}, requiring strong demand oracles for buyer values \cite{bartal2003incentive,feige2006maximizing,lavi2011truthful,dobzinski2006truthful,assadi2020improved}, or only studying the allocation problem under the assumption that all buyers' value information is public \cite{dobzinski2005approximation,dobzinski2006improved,vondrak2008optimal,lehmann2001combinatorial}. Our research investigates data auctions with a specific buyer value function that is submodular. We introduce an approximation algorithm that is truthful and runs in polynomial time. Our approach significantly simplifies \citet{dughmi2011convex}'s algorithm, which faces technical difficulties related to the solvability of convex programs.



\section{Model and Problem Statement}\label{sec:model}
Here we introduce our model, notation, and the formal statement of our problem. The seller owns a database of points $\db$ that consists of a collection of $n$ many points, i.e, $\db=\{x_1,x_2,\dots,x_n\}$. The points in the database could represent images, documents, or even articles that contain both textual and visual information that can be used for RAG. Our framework is general and agnostic to the nature of the points only requiring a measure of distance $d(.,.)$ or similarity $s(.,.)$ that can be calculated given any pair of points $x_a$ and $x_b$. The choice for the distance/similarity measure is largely dependent on the setting, but there are many well-known standard measures that can be used. For example, if the database $\db$ consists of a collection of images then a simple and meaningful measure would be the Euclidean distance between the embeddings of the images, i.e., $d(x_a,x_b)=\norm{f(x_a)-f(x_b)}_2$ where $f(.)$ is some embedding function \cite{dai2020convolutional,jeon2020acoustic}. If we had a collection of documents on the other hand, then we may use the cosine similarity between the term frequency vectors of the points $x_a$ and $x_b$, i.e., $s(x_a,x_b)=\frac{\vv{t_a} \ \cdot \ \vv{t_b}}{\norm{\vv{t_a}} \norm{\vv{t_b}}}$ where $\vv{t_a}$ and $\vv{t_b}$ are the term frequency vectors of documents $x_a$ and $x_b$, respectively \cite{harris1954distributional}. For consistency, we assume that we have a distance $d(.,.)$ instead of a similarity measure $s(.,.)$, although our framework, theoretical results, and algorithms hold for both settings. 

\textbf{Coverage-Based Value Function:} There are $m$ buyers (bidders)\footnote{Since the buyers will submit bids, we will use the words buyers and bidders interchangeably.} who are interested in (parts of) the database $\db$ and want to use it in particular for applications such a RAG or standard information retrieval. It is critical to model a value function that accurately represents the value a bidder would extract from a subset of points in the database without requiring the training of a large machine learning model as the training process itself can be quite expensive.\footnote{The cost of training large language models (LLMs) has surged in recent years, with expenses now reaching into the millions of dollars for state-of-the-art models \cite{cottier2024rising}.} Since we are motivated by RAG-based applications our value function is focused on measuring the approximate \emph{coverage} a bidder obtains from a given subset of points (dataset) from the database. To fully characterize the coverage function, we need to introduce some notation. First, consider a bidder $i$ and dataset of points $\bundle \subset \db$ that has been given to him. 
The absence of a point $j$ from $\bundle$ does not imply that it is not covered by the dataset $\bundle$ since there may exist a point $j' \in \bundle$ that is \emph{sufficiently close} to $j$ and can practically serve as a representative of $j$. However, this threshold of closeness can be dependent on the point and the bidder. Therefore, we define the coverage variable $C_{ij}$ and set $C_{ij}=1$ only if $\exists j' \in \bundle$ such that $d(j,j') \leq r_{ij}$ where $r_{ij}$ is the \emph{radius} threshold for point $j$ and bidder $i$, we set $C_{ij}=0$ otherwise. Note that our model enables a fine-grained specification of the radius $r_{ij}$, i.e., at the level of a specific bidder and a specific point. $r_{ij}$ determines the desired accuracy level buyer $i$ wants for point $j$. Practically, the radius value $r_{ij}$ can be set according to the topic (such as politics or sports) point $j$ belongs to and the desired level of accuracy of bidder $i$ for covering that topic.     

Furthermore, since each point  provides a different measure of value which can be bidder dependent, we associate a value $w_{ij} \ge 0$ for each point $j \in \db$ and each bidder $i \in [m]$. Similar to the radius, our model enables the weights to be specified at the level of the point and the buyer. In a practical setting, each buyer might be interested in a specific topic and accordingly topics which the buyer finds more relevant will receive higher weights. Based on the above we define the total coverage bidder $i$ gains from dataset $\bundle$ to be $\cov_i$ where we have 
\begin{align}\label{eq:cov_i_def}
  \textstyle  \cov_i(\bundle) = \sum_{j \in \db} w_{ij} \cdot C_{ij}
\end{align}
 It is straightforward to verify that this coverage function is submodular -- i.e., $\cov_i(\bundle_1 \cup \{x\}) - \cov_i(\bundle_1) \ge  \cov_i(\bundle_2 \cup \{x\}) - \cov_i(\bundle_2)$ for all $\bundle_1 \subseteq \bundle_2 \subseteq \db$ and $x \notin \bundle_2$. Without loss of generality we assume that $\sum_{j \in \db} w_{ij}=1$\footnote{Given any collection of non-negative coefficients $\{w_{ij}\}_{j \in \db}$ for a bidder $i$ we can simply set $w_{ij} := \frac{w_{ij}}{\sum_{j \in \db} w_{ij}}$. Note that we are implicitly assuming that $\sum_{j \in \db} w_{ij}>0$ which is reasonable as otherwise this would imply that the bidder has no interest in the whole database.}, note that this implies   $0 \leq \cov_i(\bundle) \leq 1$ for any dataset $\bundle$. 

The full value a bidder $i$ gains from dataset $\bundle$ is $v_i(\bundle)$ which is 
\begin{align}
    v_i(\bundle) = \theta_i \cdot  \cov_i(\bundle)
\end{align}
where $\theta_i \ge 0$ is a number that indicates the \emph{importance} of the coverage to the bidder. Any possible value $\theta_i \in \bidset = \{\theta_1, \cdots, \theta_{|\Lambda|}\}$ where $\bidset$ is a public set. 


\paragraph{Public and Private Parameters:} Based on the above, each bidder $i \in [m]$ is characterized by three sets of parameters $\{w_{ij}\}_{j \in \db}, \{r_{ij}\}_{j \in \db}$ and $\theta_i$. We assume that the weight $\{w_{ij}\}_{j \in \db}$ and radius parameters $\{r_{ij}\}_{j \in \db}$ are known publicly. 
We make this assumption because it is likely possible to estimate $\{w_{ij}\}_{j \in \db}$ and $\{r_{ij}\}_{j \in \db}$.
Specifically, the seller, through querying an existing Gen AI product, would likely be able to infer the distance $\{r_{ij}\}_{j \in \db}$ for which an article is considered acceptable in a RAG context for a bidder.
Similarly, by understanding the specialty of the bidder in the marketplace, the seller is likely to be able to estimate the relative value of different regions of the embedding space, i.e., the weights $\{w_{ij}\}_{j \in \db}$.
However, we generally expect that it will be impossible for the seller to estimate the bidder's overall value for RAG data $\theta_i$ and therefore assume that it is private.
We acknowledge that these assumptions may not always hold, but it is well known that the multi-dimensional mechanism design problem is generally intractable.
By assuming that only $\theta_i$ is private, we can formulate the problem as a single dimensional mechanism design problem.


\paragraph{Bidding Setting and Objectives:} Each bidder $i \in [m]$ will report a value $b_i$ privately. Once all bids are submitted, the seller will decide an \emph{allocation} and a \emph{payment}. The allocation is simply the dataset $\bundle_i \subset \db$ that will be given to bidder $i$ while the payment is the price $p(\bundle_i)$ of dataset $\bundle_i$.\footnote{When the price is specified at the level of each point $j \in \db$ the price of a dataset $\bundle$ is simply the sum of the prices of the points it contains, i.e., $p(\bundle)=\sum_{j \in \bundle} p(j)$} As noted earlier a point may not be assigned to more than one bidder, thus $\forall i,i' \in [m]$ we have $\bundle_i \cap \bundle_{i'}=\emptyset$ if $i \neq i'$. The above implies that the utility a bidder $i$ obtains is 
    $u_i(\bundle_i) = v_i(\bundle_i) - p(\bundle_i).$

As in standard mechanism design, the seller wants allocation and payment rules that lead to an \emph{incentive compatible} mechanism, i.e., a mechanism where truthful bidding is a dominant strategy for each bidder. Formally, if $b_i$ denotes bidder $i$'s bid and $\bidprof_{-i}$ denotes all other bids, then $u_i(\theta_i,\bidprof_{-i}) \ge u_i(b_i',\bidprof_{-i}) \ \forall b_i' \text{ and } \forall \bidprof_{-i}$. In other words, a bidder never gains higher utility by reporting a value other than $\theta_i$ regardless of all other bids. Note that we assume that the bid values $b_i$ like the true values $\theta_i$ are restricted to the same set $\bidset$. 

Besides incentive compatibility, given a bidding profile $\bidprof$ the seller wants to maximize the social welfare defined as 
    $\welf(\bidprof) = \sum_{i \in [m]} b_i \cdot \cov_i(\bundle_i).$
Note that the welfare is the sum of the total utilities (including that of the seller) but it excludes the payment values since the seller's utility is the sum of the payments across all bidders $\sum_{i \in [m]} p(\bundle_i)$ and therefore the payment cancels out. We will show in Section \ref{sec:new_algorithmic} that welfare maximization in our setting is NP-hard, therefore finding a welfare maximizing allocation is not possible to do in polynomial time unless P=NP. Accordingly, we aim at $\alpha$-optimal incentive compatible mechanisms which (in addition to incentive compatibility) for a given bid profile $\bidprof$ lead to a welfare value of at least $\alpha \welf^*(\bidprof)$ where $\welf^*(\bidprof)$ is the optimal welfare value for bid profile $\bidprof$. 


Additionally, we note that an important property that we target in designing our mechanism is \emph{individual rationality} which states that any truthful bidder should obtain non-negative utility regardless of the other bids, i.e., $u_i(\theta_i,\bidprof_{-i}) \ge 0 \ , \forall \bidprof_{-i}$. Moreover, if we use randomized algorithms, then we consider the expected value of the utilities and the welfare.  

Finally, in Appendix \ref{app:revenue}, we show that an $\alpha$ approximation welfare maximizing algorithm can be used to guarantee an $\alpha$ approximation for the seller's revenue $\sum_{i \in [m]} p(\bundle_i)$. In fact, if the welfare maximizing algorithm is incentive compatible and individually rational then so would the revenue maximizing algorithm. Therefore, for simplicity we focus on discussing only welfare maximization for the rest of the paper since revenue maximization is immediately implied.


\section{Welfare Maximization: Computational Hardness and Approximation}\label{sec:new_algorithmic}

Our objective is to maximize the social welfare $\welf(\bidprof)$. Unfortunately, we show that the problem is actually NP-hard. Interestingly, all of our hardness results hold even when the weights are constant, i.e., $w_{ij}=w >0 , \ \forall j \in \db, \forall i \in [m]$. However, before we show that welfare maximization is NP-Hard, we note that a common thread in combinatorial auctions has considered using the demand oracle (or query) \cite{dobzinski2005approximation,khot2005inapproximability,feige2006approximation,assadi2020improved} which is defined as follows:
\begin{definition}\label{def:demand_oracle}
Given a price vector $\priceprof$ over the points in the database and a bidder $i \in [m]$ the demand oracle returns the subset of points that maximize utility for bidder $i$, i.e., $\argmax\limits_{\bundle \subset \db} \big( v_i(\bundle)-p(\bundle) \big)$. 
\end{definition}

The demand oracle can be used to design incentive compatible auctions with high welfare. In fact, the recent work of \citet{assadi2020improved} achieves an approximation ratio of $O(\frac{1}{(\log \log n)^3})$ having a small dependence on the number of items. However, in some settings the implementation of the demand oracle might require solving an NP-hard problem \cite{dobzinski2006improved}, which render demand oracle based algorithms of little use. The proposition below shows that this is   the case in our setting as well.

\begin{proposition}\label{th:demand_oracle_is_np_hard}
The demand oracle cannot be implemented in polynomial time unless $P=NP$.
\end{proposition}

We now study welfare maximization allocation. Notably, if there is only a single bidder, welfare maximization is trivial -- just allocate all items to this bidder would maximize revenue. Surprisingly, this welfare maximization problem immediately becomes NP-hard even when there are $ 2$ bidders. 

\begin{theorem}\label{thm:welfare-hard}
Welfare maximization is NP-hard even in the following two special situations:
\begin{itemize}[leftmargin=2em] \vspace{-2mm}
  \item There are only two bidders and the radius values in the value functions differ across bidders but do not differ across data points (i.e.,  $\forall i \in [m]: r_{ij}=r_i$).  \vspace{-2mm} 
    \item There are only three bidders and their value functions   have the same radius  (i.e.,  $ r_{ij}=r \, \forall i,j$). 
\end{itemize}
\end{theorem}

The proofs of both hardness results are non-trivial, and are novel to the best of our knowledge. This is because previous hardness of submodular welfare maximization reduction does not apply to our setting since our valuation function is more special than general coverage functions, hence its intractability is more difficult to prove. The two   claims in  Theorem \ref{thm:welfare-hard} are reduced from dominating set and domatic number problem \cite{kaplan1994domatic} respectively; details are deferred to Appendix \ref{app:proof_welfare_NP_hard_m2} and \ref{app:proof_welfare_NP_hard_m3}. 

Despite the intractability of welfare maximization even with few bidders, next we introduce a polynomial time approximation algorithm that computes a $(1- 1/e)$-optimal solution for \emph{any} number of bidders. Before we show our algorithm, we introduce a convenient notation:  for a point $j \in \db$ and bidder $i \in [m]$ we denote by $N_i(j)$ the set of points in $\db$ that are within a distance of $\rad_{ij}$ from point $j$. Formally, $N_i(j)=\{j' \in \db| d(j,j') \leq \rad_{ij}\}$.

The starting point of our algorithm is a linear program (LP) relaxation of the allocation problem as shown below. 
We use the symbols $\bX_{ij}$ and $\bC_{ij}$ with a bar for the LP decision variables to emphasize that they are fractional (i.e., not necessarily equal to $0$ or $1$, but in $[0,1]$). 
\begin{subequations}  \label{opt:lp}  
\begin{equation} \label{lp:obj}  
\textstyle \max\limits_{\bX_{ij},\bC_{ij}}  \quad \sum_{i \in [m]} \theta_i \cdot \Big( \sum_{j \in  \db} w_{ij} \cdot \bC_{ij} \Big) 
\end{equation}    
\begin{equation}   \label{lp_const:between_0and1}
\textstyle \forall i\in [m], \forall j \in \db: \quad  0 \leq \bX_{ij} , \bC_{ij} \leq 1; \quad \bC_{ij} \leq \sum_{j' \in N_i(j)} \bX_{ij'} 
\end{equation}
\begin{equation}   \label{lp_const:to_atmost_onebidder}
\textstyle \forall j \in \db: \quad  \sum_{i \in [m]} \bX_{ij} \leq 1 
\end{equation}
\end{subequations}
In the LP, $\bX_{ij}$ represents the (relaxed) fractional allocation of point $j$ to bidder $i$ whereas $\bC_{ij}$ indicates whether the point is covered for bidder $i$ or not. Note that a point $j$ may be covered for a bidder $i$ even if it was not allocated to him since the bidder may have been allocated another point $j' \in N_i(j)$. Since the welfare-optimal allocation is a feasible solution to this LP,  we have the following.
\begin{lemma}\label{lem:opt_lp}
For any given bid profile $\bidprof$ the optimal value of LP \eqref{opt:lp} satisfies $\optlp(\bidprof) \ge \welf^*(\bidprof)$. 
\end{lemma}

Our full algorithm $\algone$ (Algorithm \ref{alg:lp_and_round}) is simple. We start by solving LP \eqref{opt:lp}, then round the resulting values $\bX_{ij}$ as follows, for a point $j \in \db$ we allocate it to bidder $i$ with probability $\bX_{ij}$. Note that this rounding is possible since by constraint \eqref{lp_const:to_atmost_onebidder} we have $\sum_{i \in [m]} \bX_{ij} \leq 1 $, further by constraint \eqref{lp_const:between_0and1} we have $\bX_{ij} \ge 0$. Note further that the sampling is also done in an \emph{independent}  manner from point to point. It follows that the resulting integer values $X_{ij}$ satisfy $\Prob(X_{ij}=1) = \bX_{ij}$.

\begin{algorithm}
\caption{LP + Rounding (\textbf{\textsc{LPR}})}
\label{alg:lp_and_round}
\begin{algorithmic}[1]
\INPUT Database $\db$, bidders $i \in [m]$, bidders' parameters $\{w_{ij},r_{ij}\}_{j\in \db, i \in [m]}$ and$\{\theta_i\}_{i \in [m]}$. 
\STATE Run LP \eqref{opt:lp} and obtain the optimal solution $\bX_{ij}$ and $\bC_{ij}$.  
\STATE For a given point $j \in \db$ assign it to bidder $i$ with probability $\bX_{ij}$.  
\STATE Form the dataset allocation $(\bundle_1,\bundle_2,\dots,\bundle_m)$ using the values $X_{ij}$. 
\end{algorithmic}
\end{algorithm}

The resulting coverage values for the bidders (and the welfare as well) are dependent on $C_{ij}$. However, the values $C_{ij}$ have a somewhat complicated dependence on the values of $X_{ij}$. We now introduce the following lemma which uses the optimal LP solution $\bX_{ij}, \bC_{ij}$ to give an explicit form for the expected value of $C_{ij}$ as well as lower bounding the expected total coverage for each bidder.  


\begin{lemma}\label{lem:rounding_approx}
Using Algorithm \ref{alg:lp_and_round} with optimal solution $\bX_{ij}$ and $\bC_{ij}$ the final rounded solution satisfies 
\begin{align}
    \textstyle \E[C_{ij}]    & = 1 - \prod_{j' \in N(j)} (1-\bX_{ij'}) \label{eq:point_cov_exact_form}  \\
    \E[\cov_{i}]  & \ge (1-1/e) \cdot  \bcov_{i} \label{eq:point_cov_lb} 
\end{align}
where $\bcov_i = \sum_{j \in \db} w_{ij} \bC_{ij}$. 
\end{lemma}

Since the above lemma provides  a lower bound for the coverage each bidder receives, it follows that we can lower bound the welfare as shown in the theorem below. 
\begin{theorem}\label{th:lp_round_approx}
Algorithm \ref{alg:lp_and_round} has an approximation ratio of $(1-1/e)$.    
\end{theorem}
\section{Incentive Compatible Welfare Maximization}\label{sec:new_ic_theory}

In this section, we turn to the study of incentive compatibility of welfare maximization, i.e., designing a mechanism that incentivizes every buyer $i$ to report their private type $\theta_i$ \textit{truthfully} for the data auction. We first present a nice property of the allocation rule of an incentive compatible mechanism. Specifically, we generalize Myerson's results \cite{myerson1981optimal} to our setting 
showing that any incentive compatible mechanism's allocation rule has to necessarily be monotonic. In addition to the necessity of the monotonicity property, we show an efficient algorithm for calculating the payment in order to satisfy both incentive compatibility and individual rationality. First, we start with a formal definition of a monotonic allocation rule. 

\begin{definition}\label{def:mono}[Monotonic Allocation Rule]
An allocation is monotonic if for every bidder $i \in [m]$ and bids $\bidprof_{-i}$ then we have $\E[\cov_i(\bundle)] \ge \E[\cov_i(\bundle')]$ if $b_i \ge b'_i$ where $\bundle$ and $\bundle'$ are the datasets allocated to bidder $i$ under bid profiles $(b_i,\bidprof_{-i})$ and $(b'_i,\bidprof_{-i})$, respectively. 
\end{definition}
A fundamental result of Myerson \cite{myerson1981optimal} reduces a mechanisms' truthfulness to monotonicity in allocation probability in single-item auctions. To handle our combinatorial allocation case, we start by showing a novel form of monotonicity for our combinatorial auction with coverage-based valuations.  
\begin{lemma}\label{th:gen_myerson_lemma}
For any incentive compatible mechanism, its allocation rule must be monotonic. Conversely, given a monotonic allocation rule, there always exists a corresponding efficiently computable payment such that truthful bidding is a dominant strategy and individually rational for each bidder. 
\end{lemma}
Lemma \ref{th:gen_myerson_lemma} offers a characterization of  incentive compatible mechanisms for combinatorial allocation  with our coverage-based value function.  This differs from and in some sense strictly generalizes the classical monotonicity of allocating a single item, so we believe our lemma could be of independent interest. A relevant result is by  \citet{nisan2007computationally} who also show a  generalization of Myerson's lemma to combinatorial auctions, but with   single-minded value (i.e., the buyer either has value $v$ if getting her demanded set, or value $0$ otherwise). However, our valuation function here is fundamentally different.

At this point, one would hope that our \algone{} algorithm (Algorithm \ref{alg:lp_and_round}) from the previous section would be monotonic, hence truthfulness is directly achieved. Unfortunately,    monotonicity turns out to be a very delicate property for allocation algorithms. Indeed, we find that  almost all natural allocation approaches for submodular functions  (e.g., greedy \citep{bilmes2022submodularity,lehmann2001combinatorial} and continuous greedy   \cite{vondrak2008optimal}) are \emph{not} monotone; neither is our \algone{} algorithm. 


\begin{fact}\label{thm:lp_round_not_monotone}
\algone{} (Algorithm \ref{alg:lp_and_round}) is not monotonic in the sense of Definition \ref{def:mono}.
\end{fact} 

The above fact is shown by constructing a counter example. However, the core reason behind the failure is that the relation between the rounded coverage variables $C_{ij}$ that decides the final coverage value and the LP coverage variables $\bC_{ij}$ is complicated. More specifically, while the LP solution value in the constructed example is actually monotonic, due to rounding the final expected value for the coverage based on the rounded coverage variables $C_{ij}$ is not monotonic.  


The crux of our approach is that, although Algorithm \ref{alg:lp_and_round} is not monotonic, we are able to make it monotonic through careful post-processing.  Specifically, it can be shown that the optimal coverage values $\bcov_1,\dots,\bcov_m$ each bidder receives from the optimal LP solution before rounding are, in fact, monotonic (see Lemma \ref{lemma:lp_is_monotonic}). As a result, we propose $\algtwo$ (Algorithm \ref{alg:lp_and_round_monotonic}) which utilizes a  post-processing of \emph{data burning}. That is, instead of directly allocating the resulting datasets from Algorithm \ref{alg:lp_and_round} to each bidder, the bidder is allocated the dataset with probability $\pval_i$ and given nothing with probability $1-\pval_i$. For each bidder $i\in [m]$ the probability is set to $\pval_i=(1-1/e) \cdot \frac{\bcov_i}{\E[\cov_i]}$, this ensures both that the probability value is valid (i.e., $\pval_i \in [0,1]$) and that each bidder receives an expected coverage of exactly $(1-1/e) \cdot \bcov_i$. Since for a bidder $i$ the LP coverage value $\bcov_i$ is monotonic it follows that if a bidder instead receives in expectation a positive constant of $(1-1/e)$ times the LP coverage value then the final coverage would also be monotonic. Further, while it is true that allocating nothing to a bidder with probability $1-\pval_i$ would degrade the coverage for bidder $i$, the degradation is still within a factor of $(1-1/e)$ of the LP coverage value and we can thus show that this method leads to the same approximation factor of $(1-1/e)$. Therefore, we can establish the following theorem.

\begin{algorithm}
\caption{LP + Monotonic Rounding (\textbf{\textsc{LPRMono}})}
\label{alg:lp_and_round_monotonic}
\begin{algorithmic}[1]
\INPUT Database $\db$, bidders $i \in [m]$, bidders' public parameters $\{w_{ij},r_{ij}\}_{j\in \db, i \in [m]}$, bidders' privately reported bids $\{\theta_i\}_{i \in [m]}$. 
\STATE Run LP \eqref{opt:lp} and let the output be $\bX_{ij}$ and $\bC_{ij}$.  
\STATE For a given point $j \in \db$ assign it to bidder $i$ with probability $\bX_{ij}$.  
\STATE Form the dataset allocation $(\bundle_1,\bundle_2,\dots,\bundle_m)$ using the values $X_{ij}$. 
\STATE For each bidder $i \in [m]$: \\ 
         \ \ -compute $\bcov_i =  \sum_{j \in \db} w_{ij} \cdot \bC_{ij}$. \\ 
         \ \ -compute $\E[\cov_i] = \sum_{j \in \db} w_{ij} \E[C_{ij}]$ using \eqref{eq:point_cov_exact_form}. \\
         \ \ -if $\bcov_i>0$ set $\pval_i = (1-1/e) \cdot \frac{\bcov_i}{\E[\cov_i]}$ else set $\pval_i=0$.
\STATE Each bidder $i \in [m]$ with probability $\pval_i$ gets dataset $\bundle_i$ and otherwise gets nothing. 
\end{algorithmic}
\end{algorithm}

\begin{theorem}\label{thm:lp_round_monotone}
Algorithm \ref{alg:lp_and_round_monotonic} is monotonic and has an approximation ratio of $1-1/e$.
\end{theorem}

We defer the formal proof of this theorem to Appendix \ref{app:proof_lp_round_monotone}. To illustrate its proof ideas, we present the two key lemmas below.  
\begin{lemma}\label{lemma:lp_is_monotonic}
The optimal LP allocation $\bcov_i$ resulting from LP \eqref{opt:lp} is monotonic for every bidder.
\end{lemma}
Then we show that the monotonic allocation rule also guarantees a $(1- 1/e)$ factor of the LP coverage for every bidder. Denoting the  the final coverage resulting from Algorithm \ref{alg:lp_and_round_monotonic} with a prime $'$, i.e., $\cov'_i$ we have the following lemma.  
\begin{lemma}\label{lem:expected_ceverage}
The expected coverage for any bidder $i \in [m]$ from Algorithm \ref{alg:lp_and_round_monotonic} is $\E[\cov'_i]= (1-1/e) \cdot \bcov_i$.
\end{lemma}
We believe our $(1-1/e)$ algorithm, as well as our novel monotonic post-processing technique, are of independent interest. In fact, we show that our techniques can be modified to lead to a $(1-1/e)$ welfare maximization for another variant of coverage valuations as defined in \citet{dobzinski2006improved}\footnote{The coverage valuation defined in \citet{dobzinski2006improved} is different from ours and does not generalize our coverage valuation. Nevertheless, we show that our techniques can be applied there as well.}, we can in addition use a similar post-processing method to satisfy monotonicity for these coverage valuations, see Appendix \ref{app:alg_standard_coverage} for details.  It remains an intriguing open question whether one can design an algorithm that outperforms the $1 - 1/e$ approximation for our RAG-based coverage function. While this bound is known to be tight for general coverage functions~\cite{feige1998threshold,dobzinski2005approximation}, the hope here lies in exploiting the distance-based membership relation in the RAG setting. However, tightly pining down the best ratio is beyond the scope of this paper.
 
\section{Experiments}\label{sec:experiment}

We evaluate the quality and performance of our algorithm for computing the expected truthful auction mechanism. Specifically, we consider three datasets for the selling process. While the first dataset is synthetic, the second and third are standard  text and image datasets which we expect to have similar performance to RAG settings that use information such as text or images. 
The parameters for a given setting are $m$ which is the number of buyers, $n$ which is the size of the dataset $\db$, $\Lambda$ which is the set of possible types for each buyer, and $d$ which is the dimension of the space.
Further, we always use the Euclidean distance for all instances and datasets.  

\begin{itemize}[leftmargin=*]
    \item \textbf{Synthetic} (\textbf{\textsc{Syn}}) dataset: 
    We set $m=2$, $n=5$, $|\Lambda|=10$, $d=10$. And the database $\db$ is generated uniformly within the range of $[0, 10]$ with a shape of $n \times d$. 
    This is a relatively small toy dataset aimed to test all algorithms' performances, as we will show later, one benchmark algorithm is very computationally expensive and cannot solve large datasets efficiently. 
    \item \textbf{20\_newsgroups} (\textbf{\textsc{News}}) dataset \cite{SetFit_20_newsgroups_2025}: This is a popular collection of newsgroup documents used for text classification and machine learning tasks. The training dataset contains 11,314 newsgroup posts partitioned nearly evenly across 20 different classes. Then we generate the database $\db$ by encoding the dataset with a powerful and efficient sentence embedding model all-MiniLM-L6-v2 \cite{reimers-2019-sentence-bert}, which produces an embedding dataset with a shape of $n=11,314$ and  $d=384$. In addition, we set $m=3$ and $|\Lambda| =10$. 
    \item \textbf{CIFAR-10} (\textbf{\textsc{CIFAR}}) dataset \cite{uoft-cs_cifar10_2025}: This is a widely used benchmark in machine learning and computer vision for image classification tasks. The training dataset contains 50,000 images in 10 classes, with 5,000 images per class. We downloaded the encoded embeddings from \citet{MatthewWilletts2025Cifar10} with a shape of $n=50,000$ and $d=2048$. We also set $m=3$ and $|\Lambda| =10$. 
\end{itemize}

For all three datasets, we set $\Lambda = \{0, 0.1, \cdots, 0.9\}$ 
and then each buyer's type $\theta_i$ where $i \in [m]$ is drawn independently and uniformly from $\Lambda$. In addition, we generate $r_{ij} = \alpha_i R_{jc}$ for all $i \in [m]$ and $j \in [n]$, where $\alpha_i$'s are uniformly randomly generated  coefficients within range of $[0,1]$ and $R_{jc}$ is a class dependent average distance. Note that $\news$ and $\cifar$ datasets come with different classes, and we assume that the $\syn$ dataset only contains a single class. $R_{jc}$ is computed by first sampling 100 random points within each class (all points for $\syn$) and then computing the average Euclidean distance among those 100 points. Finally, we assume $w_{ij}$ is the same for all points in the same class, which is generated uniformly randomly within the range of $[0,1]$. Then we normalize the randomly generated numbers such that $\sum_j w_{ij} = 1, \forall i \in [m]$. 


We compare our Algorithm \ref{alg:lp_and_round} ($\algone$) and Algorithm \ref{alg:lp_and_round_monotonic} ($\algtwo$) with the following allocation algorithms: 
    (1) The widely adopted  greedy ($\rg$) allocation algorithm, as used in various recent works such as \citep{agarwal2024learning,bilmes2022submodularity,lehmann2001combinatorial}; 
    (2) The optimal approximation algorithm for submodular welfare maximization without incentive constraints, termed as continuous greedy algorithm ($\cg$) that is proposed by \citet{vondrak2008optimal}. 
These algorithms are widely recognized for their effectiveness in handling the complex valuations of item bundles, which directly aligns with our problem of valuing datasets in RAG tasks. 
We note that both $\rg$ and $\cg$ are designed for the allocation problem under the assumption that buyers’ valuations are publicly known, achieving approximation ratios of $1/2$ and $1-1/e$, respectively \citep{agarwal2024learning,bilmes2022submodularity,lehmann2001combinatorial,vondrak2008optimal}. However, \textit{neither} algorithm is guaranteed to satisfy \textit{incentive compatibility}, since they rely on the assumption of public value information.
It’s also worth mentioning that \citet{dobzinski2005approximation} proposed an incentive-compatible allocation algorithm based on matching techniques, though it was developed for a different problem setting.
Their algorithm focuses on the scenario where the number of items is smaller than the number of buyers. This algorithm achieves an approximation ratio of $1/\sqrt{n}$, is monotonic, and assigns one item to one buyer. However, it is not suitable for our scenario, where the number of items is much larger than the number of buyers, and is therefore excluded from the comparison.


We compare all algorithms' performance in terms of the welfare, measured as the total utility achieved by all buyers, as well as the runtime.

First, we start by comparing the algorithms' performance on welfare. For $\algone$ and $\algtwo$, we run the rounding process (i.e., Lines 2-3 of Algorithm \ref{alg:lp_and_round} and Lines 4-5 of Algorithm \ref{alg:lp_and_round_monotonic}) for $150,000$ trials to get an estimation of the expected welfare. In addition, even though the optimal welfare of the allocation problem is NP-Hard to compute, we have the objective value from LP \eqref{opt:lp} as an upper bound on the optimal welfare. As a result, we compare each algorithm's approximation ratio with respect to the objective value of LP \eqref{opt:lp}. 

\begin{table}[htb]
    \centering
    \begin{tabular}{|c|c|c|c|}
    \hline
      Algorithm  & $\syn$ & $\news$ & $\cifar$\\
    \hline
         $\algone$& \makecell{$0.9924$ \\ $\pm 0.037$} & \makecell{$\mathbf{0.9947}$ \\ $\pm 0.014$}  & \makecell{$\mathbf{0.9982}$ \\ $\pm 0.004$}\\
    \hline
        $\algtwo$ & \makecell{$0.6319$ \\ $\pm 0.001$} & \makecell{$0.6325$ \\ $\pm 0.002$} & \makecell{$0.6322$ \\ $\pm 0.001$} \\
    \hline
        $\rg$ & \makecell{$\mathbf{0.9945}$ \\ $\pm 0.023$} & \makecell{$0.9908$ \\ $\pm 0.014$}  & \makecell{$0.9871$ \\ $\pm 0.014$} \\
    \hline
        $\cg$ & \makecell{$0.915$ \\ $\pm 0.136$} & N/A & N/A \\
    \hline
    \end{tabular}
    \caption{Welfare approximation ratio relative to objective, averaged over 50 instances with standard deviation.}
    \label{tab:welfare}
\end{table}

From the experimental results in Table \ref{tab:welfare}, we can see that $\rg$ performs slightly better than our $\algone$ on the small $\syn$ dataset, but is outperformed by $\algone$ when the size of dataset gets larger. In addition, we note that our $\algtwo$ indeed achieves a tight $1 - 1/e$ approximation ratio, as demonstrated by the experimental results. Finally, we would like to mention that $\cg$ has a runtime complexity of $O\big((mn)^8\big)$ \cite{vondrak2008optimal}, making it impractical for larger datasets like $\news$ and $\cifar$. Consequently, we exclude $\cg$ from welfare comparisons in those two datasets and the following runtime comparison. 

\begin{figure}[htb]
    \centering
    \includegraphics[width=0.6\columnwidth]{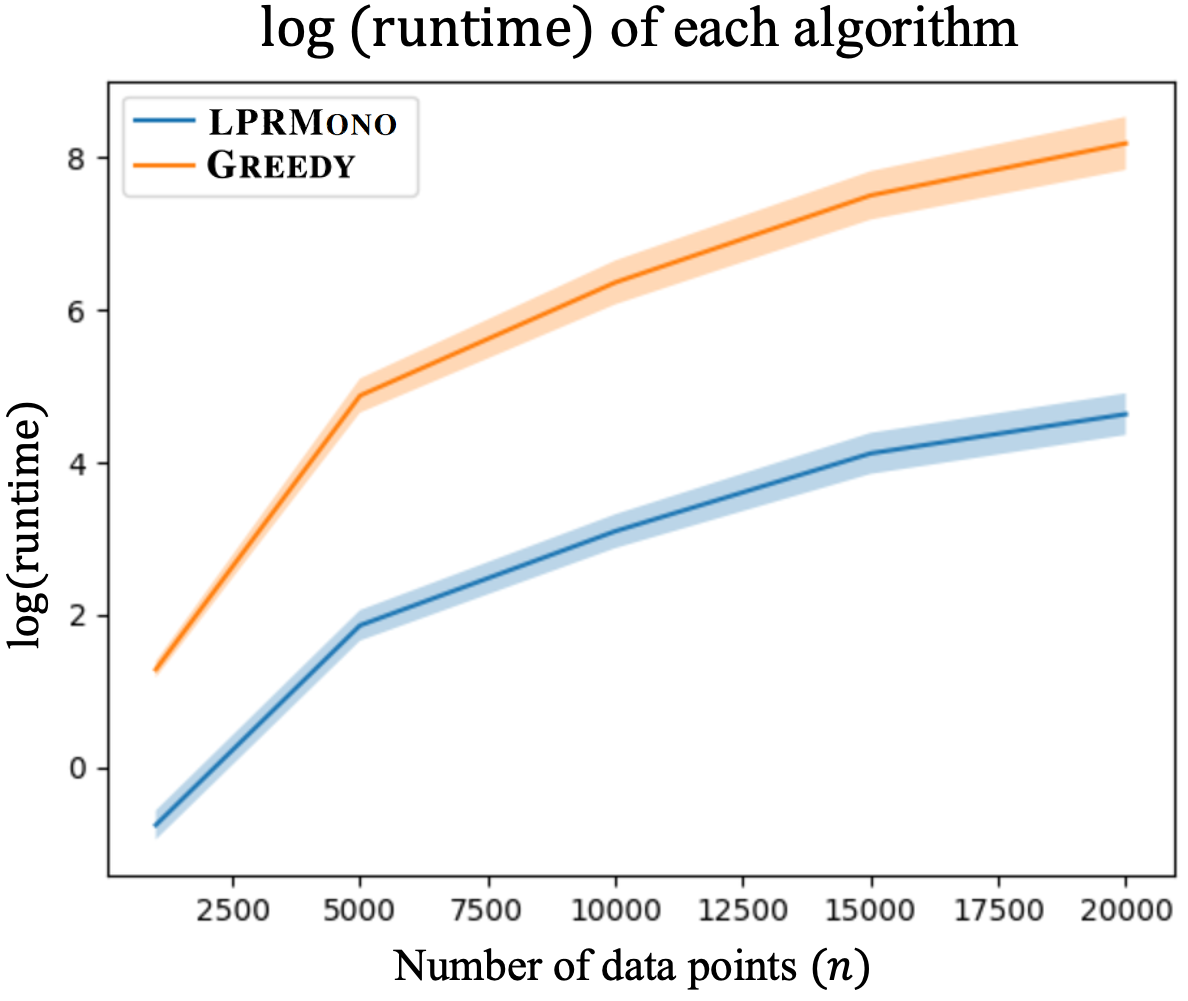}
    \caption{Runtime in seconds of each algorithm.}
    \label{fig:runtime}
\end{figure}


The runtime results is shown by Figure \ref{fig:runtime}. The difference between $\algone$ and $\algtwo$ is negligible so only $\algtwo$ is presented.  We vary the number of data points $n$ and evaluate the runtime of each algorithm. All experiments are repeated on 50 random instances on $\cifar$, and the plotted curves are averaged across the runs with error bars for 95$\%$ confidence level. We can see that the runtime of $\rg$ is significantly larger compared to our algorithms. 
It is also worth noting that besides the runtime advantage of our algorithms, $\algtwo$ is incentive compatible. 
Finally, we show that $\rg$ is not monotone with the following Example. 

\begin{example}[Counterexample Demonstrating the Non-monotonicity of $\rg$]\label{tab:not_monotone_example}
\textit{Consider a $\db$ with three data points on a line $A - B - C$, where the distances between $(A,B)$, $(A,C)$, and $(B,C)$ are $1, 2, 1$, respectively. We have three bidders all having the same radius of $r_{ij}=1$ and weight of $w_{ij}=\frac{1}{3}$ for all points. Further,  the bidders have $\theta_1=1$, $\theta_2=\frac{3}{4}$, and $\theta_3=\frac{1}{2}$. The greedy algorithm will loop over the data points starting from point $A$ then $B$ and finally $C$. At each iteration $t$ the marginal gain $v_i(S_{i,t}\cup  \{j\})-v_i(S_{i,t})$ is calculated for each bidder $i$ where $S_{i,t}$ is the set of points given to bidder $i$ before the $t^{th}$ iteration is calculated and the point $j$ is the one considered at iteration $t$. Point $j$ is then allocated to the bidder with the maximum gain.} 

\textit{Now, consider the truthful bid profile of $\bidprof=(b_1=1,b_2=\frac{3}{4},b_3=\frac{1}{2})$. By running the greedy algorithm we can see that the solution would result in bidder $1$ obtaining only point $A$ and hence a coverage of $\cov_1=\frac{2}{3}$. On the other hand, if bidder 1 deviates to $b'_1=\frac{7}{10}$ (with the other bids fixed), then he would instead obtain point $B$ leading to $\cov'_1=1>\cov_1=\frac{2}{3}$. This clearly violates monotonicity. }
\end{example}

It is challenging to theoretically demonstrate that the $\cg$ algorithm is not monotonic due to its complicated randomized nature which would require a lengthy manual simulation of the algorithm. Consequently, we assessed its monotonicity on a randomly generated instance from the $\syn$ dataset. Given the stochastic nature of randomized allocation algorithms, we estimated the expected allocation coverage $\mathbb{E}[\cov_i(\bundle)]$ for every buyer $i$ at a given bid value by running the algorithm 150,000 times and averaging the outcomes.\footnote{Note that since the algorithm is randomized, finding the exact expected coverage value would require infinitely many samples.} 
Even though 150,000 is a large number, our estimation would still have some minor deviations. Based on standard concentration inequalities \cite{hoeffding1994probability}, it is straightforward to show that using 150,000 runs makes our estimates accurate to a deviation of at most $\epsilon = 0.01$ with probability at least $99.99\%$. Let $\hat{\cov}_i(\bundle)$ and $\hat{\cov}_i(\bundle')$ denote our estimates of the expected coverage for buyer $i$ at bids $b_i$ and $b_i'$, respectively. Then, if we have $b_i \ge b_i'$ and obtain estimates such that  $\hat{\cov}_i(\bundle)+\epsilon < \hat{\cov}_i(\bundle') -\epsilon$, then $\cg$ is not monotonic with probability at least $99.99\%$. This in fact happens on at least one specific instance from $\syn$, the full details of this instance can be found in Appendix \ref{app_sec:counterexample}.

\bibliography{refer}

\begin{thebibliography}{56}
\providecommand{\natexlab}[1]{#1}
\providecommand{\url}[1]{\texttt{#1}}
\expandafter\ifx\csname urlstyle\endcsname\relax
  \providecommand{\doi}[1]{doi: #1}\else
  \providecommand{\doi}{doi: \begingroup \urlstyle{rm}\Url}\fi

\bibitem[Achiam et~al.(2023)Achiam, Adler, Agarwal, Ahmad, Akkaya, Aleman,
  Almeida, Altenschmidt, Altman, Anadkat, et~al.]{achiam2023gpt}
Josh Achiam, Steven Adler, Sandhini Agarwal, Lama Ahmad, Ilge Akkaya,
  Florencia~Leoni Aleman, Diogo Almeida, Janko Altenschmidt, Sam Altman,
  Shyamal Anadkat, et~al.
\newblock Gpt-4 technical report.
\newblock \emph{arXiv preprint arXiv:2303.08774}, 2023.

\bibitem[Agarwal et~al.(2019)Agarwal, Dahleh, and
  Sarkar]{agarwal2019marketplace}
Anish Agarwal, Munther Dahleh, and Tuhin Sarkar.
\newblock A marketplace for data: An algorithmic solution.
\newblock In \emph{Proceedings of the 2019 ACM Conference on Economics and
  Computation}, pages 701--726, 2019.

\bibitem[Agarwal et~al.(2024)Agarwal, Dahleh, Horel, and
  Rui]{agarwal2024towards}
Anish Agarwal, Munther Dahleh, Thibaut Horel, and Maryann Rui.
\newblock Towards data auctions with externalities.
\newblock \emph{Games and Economic Behavior}, 148:\penalty0 323--356, 2024.

\bibitem[Agarwal and Balkanski(2024)]{agarwal2024learning}
Arpit Agarwal and Eric Balkanski.
\newblock Learning-augmented dynamic submodular maximization.
\newblock \emph{Advances in Neural Information Processing Systems},
  37:\penalty0 14148--14176, 2024.

\bibitem[Asai et~al.(2023)Asai, Min, Zhong, and Chen]{asai2023retrieval}
Akari Asai, Sewon Min, Zexuan Zhong, and Danqi Chen.
\newblock Retrieval-based language models and applications.
\newblock In \emph{Proceedings of the 61st Annual Meeting of the Association
  for Computational Linguistics (Volume 6: Tutorial Abstracts)}, pages 41--46,
  2023.

\bibitem[Assadi and Singla(2020)]{assadi2020improved}
Sepehr Assadi and Sahil Singla.
\newblock Improved truthful mechanisms for combinatorial auctions with
  submodular bidders.
\newblock \emph{ACM SIGecom Exchanges}, 18\penalty0 (1):\penalty0 19--27, 2020.

\bibitem[Bartal et~al.(2003)Bartal, Gonen, and Nisan]{bartal2003incentive}
Yair Bartal, Rica Gonen, and Noam Nisan.
\newblock Incentive compatible multi unit combinatorial auctions.
\newblock In \emph{Proceedings of the 9th conference on Theoretical aspects of
  rationality and knowledge}, pages 72--87, 2003.

\bibitem[Bilmes(2022)]{bilmes2022submodularity}
Jeff Bilmes.
\newblock Submodularity in machine learning and artificial intelligence.
\newblock \emph{arXiv preprint arXiv:2202.00132}, 2022.

\bibitem[Chen et~al.(2019)Chen, Koutris, and Kumar]{chen2019towards}
Lingjiao Chen, Paraschos Koutris, and Arun Kumar.
\newblock Towards model-based pricing for machine learning in a data
  marketplace.
\newblock In \emph{Proceedings of the 2019 international conference on
  management of data}, pages 1535--1552, 2019.

\bibitem[Cottier et~al.(2024)Cottier, Rahman, Fattorini, Maslej, and
  Owen]{cottier2024rising}
Ben Cottier, Robi Rahman, Loredana Fattorini, Nestor Maslej, and David Owen.
\newblock The rising costs of training frontier ai models.
\newblock \emph{arXiv preprint arXiv:2405.21015}, 2024.

\bibitem[Dai et~al.(2020)Dai, Yan, Zhou, Wang, Yang, and
  Cheng]{dai2020convolutional}
Xinyan Dai, Xiao Yan, Kaiwen Zhou, Yuxuan Wang, Han Yang, and James Cheng.
\newblock Convolutional embedding for edit distance.
\newblock In \emph{proceedings of the 43rd international ACM SIGIR conference
  on Research and Development in information retrieval}, pages 599--608, 2020.

\bibitem[Dobzinski and Nisan(2007)]{dobzinski2007mechanisms}
Shahar Dobzinski and Noam Nisan.
\newblock Mechanisms for multi-unit auctions.
\newblock In \emph{Proceedings of the 8th ACM conference on Electronic
  commerce}, pages 346--351, 2007.

\bibitem[Dobzinski and Schapira(2006)]{dobzinski2006improved}
Shahar Dobzinski and Michael Schapira.
\newblock An improved approximation algorithm for combinatorial auctions with
  submodular bidders.
\newblock In \emph{Proceedings of the seventeenth annual ACM-SIAM symposium on
  Discrete algorithm}, pages 1064--1073, 2006.

\bibitem[Dobzinski et~al.(2005)Dobzinski, Nisan, and
  Schapira]{dobzinski2005approximation}
Shahar Dobzinski, Noam Nisan, and Michael Schapira.
\newblock Approximation algorithms for combinatorial auctions with
  complement-free bidders.
\newblock In \emph{Proceedings of the thirty-seventh annual ACM symposium on
  Theory of computing}, pages 610--618, 2005.

\bibitem[Dobzinski et~al.(2006)Dobzinski, Nisan, and
  Schapira]{dobzinski2006truthful}
Shahar Dobzinski, Noam Nisan, and Michael Schapira.
\newblock Truthful randomized mechanisms for combinatorial auctions.
\newblock In \emph{Proceedings of the thirty-eighth annual ACM symposium on
  Theory of computing}, pages 644--652, 2006.

\bibitem[Dubey et~al.(2024)Dubey, Feng, Kidambi, Mehta, and
  Wang]{dubey2024auctions}
Avinava Dubey, Zhe Feng, Rahul Kidambi, Aranyak Mehta, and Di~Wang.
\newblock Auctions with llm summaries.
\newblock In \emph{Proceedings of the 30th ACM SIGKDD Conference on Knowledge
  Discovery and Data Mining}, pages 713--722, 2024.

\bibitem[Duetting et~al.(2024)Duetting, Mirrokni, Paes~Leme, Xu, and
  Zuo]{duetting2024mechanism}
Paul Duetting, Vahab Mirrokni, Renato Paes~Leme, Haifeng Xu, and Song Zuo.
\newblock Mechanism design for large language models.
\newblock In \emph{Proceedings of the ACM on Web Conference 2024}, pages
  144--155, 2024.

\bibitem[Dughmi et~al.(2011)Dughmi, Roughgarden, and Yan]{dughmi2011convex}
Shaddin Dughmi, Tim Roughgarden, and Qiqi Yan.
\newblock From convex optimization to randomized mechanisms: toward optimal
  combinatorial auctions.
\newblock In \emph{Proceedings of the forty-third annual ACM symposium on
  Theory of computing}, pages 149--158, 2011.

\bibitem[Feige(1998)]{feige1998threshold}
Uriel Feige.
\newblock A threshold of ln n for approximating set cover.
\newblock \emph{Journal of the ACM (JACM)}, 45\penalty0 (4):\penalty0 634--652,
  1998.

\bibitem[Feige(2006)]{feige2006maximizing}
Uriel Feige.
\newblock On maximizing welfare when utility functions are subadditive.
\newblock In \emph{Proceedings of the thirty-eighth annual ACM symposium on
  Theory of computing}, pages 41--50, 2006.

\bibitem[Feige and Vondr{\'a}k(2006)]{feige2006approximation}
Uriel Feige and Jan Vondr{\'a}k.
\newblock Approximation algorithms for allocation problems: Improving the
  factor of 1-1/e.
\newblock In \emph{2006 47th Annual IEEE Symposium on Foundations of Computer
  Science (FOCS'06)}, pages 667--676. IEEE, 2006.

\bibitem[Gao et~al.(2023)Gao, Xiong, Gao, Jia, Pan, Bi, Dai, Sun, and
  Wang]{gao2023retrieval}
Yunfan Gao, Yun Xiong, Xinyu Gao, Kangxiang Jia, Jinliu Pan, Yuxi Bi, Yi~Dai,
  Jiawei Sun, and Haofen Wang.
\newblock Retrieval-augmented generation for large language models: A survey.
\newblock \emph{arXiv preprint arXiv:2312.10997}, 2023.

\bibitem[Garey and Johnson(1979)]{garey1979computers}
Michael~R Garey and David~S Johnson.
\newblock \emph{Computers and intractability}, volume 174.
\newblock freeman San Francisco, 1979.

\bibitem[Ghorbani and Zou(2019)]{ghorbani2019data}
Amirata Ghorbani and James Zou.
\newblock Data shapley: Equitable valuation of data for machine learning.
\newblock In \emph{International conference on machine learning}, pages
  2242--2251. PMLR, 2019.

\bibitem[Hajiaghayi et~al.(2024)Hajiaghayi, Lahaie, Rezaei, and
  Shin]{hajiaghayi2024ad}
MohammadTaghi Hajiaghayi, Sebastien Lahaie, Keivan Rezaei, and Suho Shin.
\newblock Ad auctions for llms via retrieval augmented generation.
\newblock In \emph{The Thirty-eighth Annual Conference on Neural Information
  Processing Systems}, 2024.
\newblock URL \url{https://openreview.net/forum?id=Ujo8V7iXmR}.

\bibitem[Han et~al.(2023)Han, Light, Xia, Galhotra, Fernandez, and
  Xu]{han2023data}
Minbiao Han, Jonathan Light, Steven Xia, Sainyam Galhotra, Raul~Castro
  Fernandez, and Haifeng Xu.
\newblock A data-centric online market for machine learning: From discovery to
  pricing.
\newblock \emph{arXiv preprint arXiv:2310.17843}, 2023.

\bibitem[Harris(1954)]{harris1954distributional}
Zellig~S Harris.
\newblock Distributional structure, 1954.

\bibitem[Ho et~al.(2020)Ho, Jain, and Abbeel]{ho2020denoising}
Jonathan Ho, Ajay Jain, and Pieter Abbeel.
\newblock Denoising diffusion probabilistic models.
\newblock \emph{Advances in neural information processing systems},
  33:\penalty0 6840--6851, 2020.

\bibitem[Ho et~al.(2022)Ho, Salimans, Gritsenko, Chan, Norouzi, and
  Fleet]{ho2022video}
Jonathan Ho, Tim Salimans, Alexey Gritsenko, William Chan, Mohammad Norouzi,
  and David~J Fleet.
\newblock Video diffusion models.
\newblock \emph{Advances in Neural Information Processing Systems},
  35:\penalty0 8633--8646, 2022.

\bibitem[Hoeffding(1994)]{hoeffding1994probability}
Wassily Hoeffding.
\newblock Probability inequalities for sums of bounded random variables.
\newblock \emph{The collected works of Wassily Hoeffding}, pages 409--426,
  1994.

\bibitem[Huang et~al.(2023)Huang, Yu, Ma, Zhong, Feng, Wang, Chen, Peng, Feng,
  Qin, et~al.]{huang2023survey}
Lei Huang, Weijiang Yu, Weitao Ma, Weihong Zhong, Zhangyin Feng, Haotian Wang,
  Qianglong Chen, Weihua Peng, Xiaocheng Feng, Bing Qin, et~al.
\newblock A survey on hallucination in large language models: Principles,
  taxonomy, challenges, and open questions.
\newblock \emph{arXiv preprint arXiv:2311.05232}, 2023.

\bibitem[Jeon(2020)]{jeon2020acoustic}
Woojay Jeon.
\newblock Acoustic neighbor embeddings.
\newblock \emph{arXiv preprint arXiv:2007.10329}, 2020.

\bibitem[Jia et~al.(2019)Jia, Dao, Wang, Hubis, Hynes, G{\"u}rel, Li, Zhang,
  Song, and Spanos]{jia2019towards}
Ruoxi Jia, David Dao, Boxin Wang, Frances~Ann Hubis, Nick Hynes, Nezihe~Merve
  G{\"u}rel, Bo~Li, Ce~Zhang, Dawn Song, and Costas~J Spanos.
\newblock Towards efficient data valuation based on the shapley value.
\newblock In \emph{The 22nd International Conference on Artificial Intelligence
  and Statistics}, pages 1167--1176. PMLR, 2019.

\bibitem[Jiang et~al.(2024)Jiang, Wang, Shen, Kim, and Kim]{jiang2024survey}
Juyong Jiang, Fan Wang, Jiasi Shen, Sungju Kim, and Sunghun Kim.
\newblock A survey on large language models for code generation.
\newblock \emph{arXiv preprint arXiv:2406.00515}, 2024.

\bibitem[Kaplan and Shamir(1994)]{kaplan1994domatic}
Haim Kaplan and Ron Shamir.
\newblock The domatic number problem on some perfect graph families.
\newblock \emph{Information Processing Letters}, 49\penalty0 (1):\penalty0
  51--56, 1994.

\bibitem[Khot et~al.(2005)Khot, Lipton, Markakis, and
  Mehta]{khot2005inapproximability}
Subhash Khot, Richard~J Lipton, Evangelos Markakis, and Aranyak Mehta.
\newblock Inapproximability results for combinatorial auctions with submodular
  utility functions.
\newblock In \emph{Internet and Network Economics: First International
  Workshop, WINE 2005, Hong Kong, China, December 15-17, 2005. Proceedings 1},
  pages 92--101. Springer, 2005.

\bibitem[Lavi and Swamy(2011)]{lavi2011truthful}
Ron Lavi and Chaitanya Swamy.
\newblock Truthful and near-optimal mechanism design via linear programming.
\newblock \emph{Journal of the ACM (JACM)}, 58\penalty0 (6):\penalty0 1--24,
  2011.

\bibitem[Lehmann et~al.(2001)Lehmann, Lehmann, and
  Nisan]{lehmann2001combinatorial}
Benny Lehmann, Daniel Lehmann, and Noam Nisan.
\newblock Combinatorial auctions with decreasing marginal utilities.
\newblock In \emph{Proceedings of the 3rd ACM conference on Electronic
  Commerce}, pages 18--28, 2001.

\bibitem[Lehmann et~al.(2002)Lehmann, O{\'c}allaghan, and
  Shoham]{lehmann2002truth}
Daniel Lehmann, Liadan~Ita O{\'c}allaghan, and Yoav Shoham.
\newblock Truth revelation in approximately efficient combinatorial auctions.
\newblock \emph{Journal of the ACM (JACM)}, 49\penalty0 (5):\penalty0 577--602,
  2002.

\bibitem[Lewis et~al.(2020)Lewis, Perez, Piktus, Petroni, Karpukhin, Goyal,
  K{\"u}ttler, Lewis, Yih, Rockt{\"a}schel, et~al.]{lewis2020retrieval}
Patrick Lewis, Ethan Perez, Aleksandra Piktus, Fabio Petroni, Vladimir
  Karpukhin, Naman Goyal, Heinrich K{\"u}ttler, Mike Lewis, Wen-tau Yih, Tim
  Rockt{\"a}schel, et~al.
\newblock Retrieval-augmented generation for knowledge-intensive nlp tasks.
\newblock \emph{Advances in Neural Information Processing Systems},
  33:\penalty0 9459--9474, 2020.

\bibitem[Liu et~al.(2021)Liu, Lou, Liu, Xiong, Pei, and Sun]{liu2021dealer}
Jinfei Liu, Jian Lou, Junxu Liu, Li~Xiong, Jian Pei, and Jimeng Sun.
\newblock Dealer: an end-to-end model marketplace with differential privacy.
\newblock \emph{Proceedings of the VLDB Endowment}, 14\penalty0 (6), 2021.

\bibitem[MatthewWilletts(2025)]{MatthewWilletts2025Cifar10}
MatthewWilletts.
\newblock Resnet15 embeddings of cifar10.
\newblock \url{https://github.com/MatthewWilletts/Embeddings}, 2025.
\newblock Accessed: 2025-01-20.

\bibitem[Mehta et~al.(2021)Mehta, Dawande, Janakiraman, and
  Mookerjee]{mehta2021sell}
Sameer Mehta, Milind Dawande, Ganesh Janakiraman, and Vijay Mookerjee.
\newblock How to sell a data set? pricing policies for data monetization.
\newblock \emph{Information Systems Research}, 32\penalty0 (4):\penalty0
  1281--1297, 2021.

\bibitem[Mu'Alem and Nisan(2008)]{mu2008truthful}
Ahuva Mu'Alem and Noam Nisan.
\newblock Truthful approximation mechanisms for restricted combinatorial
  auctions.
\newblock \emph{Games and Economic Behavior}, 64\penalty0 (2):\penalty0
  612--631, 2008.

\bibitem[Myerson(1981)]{myerson1981optimal}
Roger~B Myerson.
\newblock Optimal auction design.
\newblock \emph{Mathematics of operations research}, 6\penalty0 (1):\penalty0
  58--73, 1981.

\bibitem[Nisan and Ronen(2007)]{nisan2007computationally}
Noam Nisan and Amir Ronen.
\newblock Computationally feasible vcg mechanisms.
\newblock \emph{Journal of Artificial Intelligence Research}, 29:\penalty0
  19--47, 2007.

\bibitem[Reimers and Gurevych(2019)]{reimers-2019-sentence-bert}
Nils Reimers and Iryna Gurevych.
\newblock Sentence-bert: Sentence embeddings using siamese bert-networks.
\newblock In \emph{Proceedings of the 2019 Conference on Empirical Methods in
  Natural Language Processing}. Association for Computational Linguistics, 11
  2019.
\newblock URL \url{https://arxiv.org/abs/1908.10084}.

\bibitem[Schoch et~al.(2022)Schoch, Xu, and Ji]{schoch2022cs}
Stephanie Schoch, Haifeng Xu, and Yangfeng Ji.
\newblock Cs-shapley: class-wise shapley values for data valuation in
  classification.
\newblock \emph{Advances in Neural Information Processing Systems},
  35:\penalty0 34574--34585, 2022.

\bibitem[SetFit(2025)]{SetFit_20_newsgroups_2025}
SetFit.
\newblock 20\_newsgroups.
\newblock \url{https://huggingface.co/datasets/SetFit/20_newsgroups}, 2025.
\newblock Accessed: 2025-01-20.

\bibitem[Soumalias et~al.(2024)Soumalias, Curry, and
  Seuken]{soumalias2024truthful}
Ermis Soumalias, Michael~J Curry, and Sven Seuken.
\newblock Truthful aggregation of llms with an application to online
  advertising.
\newblock \emph{arXiv preprint arXiv:2405.05905}, 2024.

\bibitem[Sun et~al.(2024)Sun, Chen, Wang, Chen, and Deng]{sun2024mechanism}
Haoran Sun, Yurong Chen, Siwei Wang, Wei Chen, and Xiaotie Deng.
\newblock Mechanism design for llm fine-tuning with multiple reward models.
\newblock \emph{arXiv preprint arXiv:2405.16276}, 2024.

\bibitem[Tang et~al.(2021)Tang, Ghorbani, Yamashita, Rehman, Dunnmon, Zou, and
  Rubin]{tang2021data}
Siyi Tang, Amirata Ghorbani, Rikiya Yamashita, Sameer Rehman, Jared~A Dunnmon,
  James Zou, and Daniel~L Rubin.
\newblock Data valuation for medical imaging using shapley value and
  application to a large-scale chest x-ray dataset.
\newblock \emph{Scientific reports}, 11\penalty0 (1):\penalty0 8366, 2021.

\bibitem[Touvron et~al.(2023)Touvron, Lavril, Izacard, Martinet, Lachaux,
  Lacroix, Rozi{\`e}re, Goyal, Hambro, Azhar, et~al.]{touvron2023llama}
Hugo Touvron, Thibaut Lavril, Gautier Izacard, Xavier Martinet, Marie-Anne
  Lachaux, Timoth{\'e}e Lacroix, Baptiste Rozi{\`e}re, Naman Goyal, Eric
  Hambro, Faisal Azhar, et~al.
\newblock Llama: Open and efficient foundation language models.
\newblock \emph{arXiv preprint arXiv:2302.13971}, 2023.

\bibitem[uoft\_cs(2025)]{uoft-cs_cifar10_2025}
uoft\_cs.
\newblock cifar10.
\newblock \url{https://huggingface.co/datasets/uoft-cs/cifar10}, 2025.
\newblock Accessed: 2025-01-20.

\bibitem[Vondr{\'a}k(2008)]{vondrak2008optimal}
Jan Vondr{\'a}k.
\newblock Optimal approximation for the submodular welfare problem in the value
  oracle model.
\newblock In \emph{Proceedings of the fortieth annual ACM symposium on Theory
  of computing}, pages 67--74, 2008.

\bibitem[Wang et~al.(2024)Wang, Mittal, and Jia]{wang2024efficient}
Jiachen~T Wang, Prateek Mittal, and Ruoxi Jia.
\newblock Efficient data shapley for weighted nearest neighbor algorithms.
\newblock \emph{arXiv preprint arXiv:2401.11103}, 2024.

\end{thebibliography}
\bibliographystyle{plainnat}

\newpage
\appendix
\newpage 
\section{Useful Facts and Lemmas}
\begin{fact}\label{fact:lower_bound_on_exp}
$\forall x \in [0,1]$ we have 
\begin{align}
    1-e^{-x} \ge (1-1/e) x
\end{align}
\end{fact}
\begin{proof}
Following a standard approach, we define the function $f(x)=1-e^{-x}-(1-1/e) x$. Therefore, we only need to show that $f(x)\ge 0 \ \forall x \in [0,1]$. Note that $f(0)=f(1)=0$. Further, the derivative is $f'(x)=e^{-x}-(1-1/e)$. Note that $f'(x)=0$ at $x=\ln \frac{e}{e-1} \in [0,1]$. Further, by the form of the derivative $f'(x)=e^{-x}-(1-1/e)$, it is clear that in the interval $[0,1]$ the function $f$ increases from $0$ at $x=0$ until $\ln \frac{e}{e-1}$ and then decreases until it has a value of $0$ again at $x=1$. Therefore, $f(x)\ge 0 \ \forall x \in [0,1]$. 
\end{proof}

\begin{lemma}\label{lemma:gen_monotonicity}
For any function $f$ that can be written as
\begin{align}
    f(Z) = \beta_1 f_1(Z) + \beta_2 f_2(Z) + \dots + \beta_m f_m(Z)  \quad , \quad Z \in S
\end{align}
where $S$ is some set and $\beta_i > 0 \ , \forall i \in [m]$, the value of $f_i(.)$ at any optimal solution is monotonic in $\beta_i$. I.e., let $Z^*$ be the optimal solution for $f$ given $(\beta_i,\beta_{-i})$ and let  $Z'^*$ be the optimal solution for for $f$ given $(\beta_i',\beta_{-i})$, then if $\beta_i > \beta'_i$ then $f_i(Z^*) \ge f_i(Z'^*)$. 
\end{lemma}
\begin{proof}
The proof is by contradiction, so we will assume that $f_i(Z^*) < f_i(Z'^*)$ with $\beta_i > \beta'_i$. Note first the welfare under different bids can be written as 
\begin{align}
    f(Z^*) & = \beta_i \cdot f_i(Z^*) +  \underbrace{\sum_{k\in [m], k \neq i} \beta_k \cdot f_k(Z^*)}_{V} \\
    f(Z'^*) & = \beta'_i \cdot f_i(Z'^*) +  \underbrace{\sum_{k\in [m], k \neq i} \beta_k \cdot f_k(Z'^*)}_{V'} 
\end{align}

We now have the following claim 
\begin{claim}\label{cl:second}
$f_i(Z'^*)-f_i(Z^*) \leq \frac{V-V'}{\beta_i}$.  
\end{claim}
\begin{proof}
Since $Z^*$ is an optimal solution then it follows that by using the solution $Z'^*$ we would not get a higher welfare under $\beta_i$. Therefore, we have 
\begin{align}
    \beta_i \cdot f_i(Z^*) + V \ge \beta_i \cdot f_i(Z'^*) + V'
\end{align}
By re-arranging the terms the inequality follows. 
\end{proof}
Now we reach the following claim
\begin{claim}\label{cl:last}
$\beta_i' \cdot f_i(Z^*) + V > \beta_i' \cdot f_i(Z'^*) + V'$.
\end{claim}
\begin{proof}
By taking the difference we get 
\begin{align}
    \beta_i' \cdot f_i(Z^*) + V - \Big( \beta_i' \cdot f_i(Z'^*) + V' \Big) & = \beta_i' \cdot  \Big(f_i(Z^*)-f_i(Z'^*)\Big) + V -V' \\ 
    & \ge -\beta_i' \cdot \Big( \frac{V-V'}{\beta_i} \Big) + V-V'  \quad \quad \text{(This follows by Claim \ref{cl:second})} \\ 
    & =\Big(  V-V' \Big) \cdot \Big( 1 - \frac{\beta'_i}{\beta_i} \Big) \\
    & > 0 
\end{align}
Where the last step follows since by Claim \ref{cl:last} we have $V-V' >0$ since by assumption $f_i(Z'^*)-f_i(Z^*) >0$ and $\beta_i>0$. Further, since $\beta'_i <\beta_i$ then $1 - \frac{\beta'_i}{\beta_i} >0$. Therefore, the multiplication of two positive number results in a positive. 
\end{proof}
Looking at the last Claim \ref{cl:last} we see that using the solutuion $Z^*$ for $(\beta'_i,\beta_{-i})$ would result in a strictly higher value for $f(.)$ than using $Z'^*$, which contradicts the assumption that $Z'^*$ is an optimal solution for $(\beta'_i,\beta_{-i})$. Therefore, we must have $f_i(Z^*) \ge f_i(Z'^*)$. 
\end{proof}

\section{Omitted Proofs of Section \ref{sec:new_algorithmic}} 

\subsection{Proof of  Proposition  \ref{th:demand_oracle_is_np_hard}  } \label{app:proof_demand_oracle_is_np_hard}
The reduction is from the dominating set problem. In the dominating set decision problem, we are given a graph $G = (V, E)$ and an integer $k$ and we should decide if there exists a subset of vertices $V' \subset V$ such that $|V'|\leq k$ and every vertex in $V$ is either in $V'$ or has an edge connected to $V'$. The problem is known to be NP-complete \cite{garey1979computers}. 

The vertices will be the points in the database $\db$ and the distance between any pair of points (vertices) will be the path distance. We let a bidder $i$ have $r_{ij}=1 \ , \forall j \in \db$ and set $b_i=1$. We set all prices for all points to $p_j=p \ ,\forall j \in \db$ where $0 < p < \frac{1}{n}$. Note that for any non-trivial instance we have $n>k$ and therefore $\frac{1}{n} < \frac{1}{k}$. Further, the weights are the same for all points, i.e., we have $w_{ij}=\frac{1}{n} \ , \forall j \in \db$. We will prove that the graph $G$ has a dominating set of size at most $k$ if and only if there exists a subset $V' \subset V$ such that $u_i(V')\ge 1-pk$. 

The first direction is easy to show. In particular, suppose that $V'$ is a dominating set of size at most $k$ then clearly $u_i(V') \ge 1-pk$. 

For the other direction, we will show that if $G$ does not have a dominating set of size at most $k$ then there exists no subset $V'$ such that $u_i(V') \ge 1-pk$. First, if $|V'|>k$ then $u_i(S') \leq 1 - p \cdot (k+1) < 1- pk$. Therefore, we must have $|V'| = k-\ell$ where $\ell$ is a non-negative integer. We will argue that the size of the covered subset, denoted by $N(V')$, satisfies $N(V') \leq n-\ell-1$. Otherwise, $V'$ can be used to construct a dominating set of size at most $k$. To see that note, we can add the $\ell$ many vertices not included in $N(V')$ and we would cover the full graph using $k-\ell+\ell=k$ many vertices which is a dominating set of size at most $k$. Thus, we must $|V'| = k-\ell$ and $N(V') \leq n-\ell-1$. Now we upper bound the utility
\begin{align*}
    u_i(V') & \leq \frac{n-\ell-1}{n} - p \cdot (k-\ell) \\ 
          & = 1 - pk + p \ell -\frac{\ell+1}{n} \\ 
          & = 1 - pk  -\frac{1}{n} + \ell (p-\frac{1}{n})  \\ 
          & \leq 1-pk  -\frac{1}{n}  \quad \text{(since $(p-\frac{1}{n})<0$ as $p<\frac{1}{n}$)}  \\ 
          & < 1-pk  
\end{align*}

\subsection{Proof of the First Claim in Theorem \ref{thm:welfare-hard} }  \label{app:proof_welfare_NP_hard_m2}
The reduction from dominating set to our problem is as follows. We will have two bidders ($m=2$). For simpliclity, we set the bid values to $b_1 =b_2 =1$. The weights for all points and bidders are set to the same value, it follows that $w_{ij}=\frac{1}{n} \ , \forall j \in \db, \forall i\in [m]$. Each bidder has the same radius value across all points. In particular, we set $r_1=1$ and $r_2=0$. Given the graph $G = (V, E)$  from the dominating set problem, each vertex $v \in V$ will represent a point in the database $\db$. For any two points (vertices) $v,v' \in \db$ we set $d(v,v')$ to the path distance between $v, v' \in V$. Therefore, $d(v,v')=0$ only if $v=v'$. Moreover, $d(v,v')=1$ only if there exists an edge between $v$ and $v'$. For all other cases, $d(v,v') >1$. 

Now for a subset of $V' \subset V$ define $N(V')$ to be the number the vertices of $V'$ union the vertices that have an edge connected to a vertex in $V'$. It follows that if we give vertices $V_1$ and $V_2$ to the first and second bidder, respectively, then the welfare would be: 
\begin{align}
    \welf &= v_1(V_1) + v_2(V_2) \nonumber \\ 
          & = \frac{N(V_1)+|V-V_1|}{n} \label{eq:domset_plug}
\end{align}
Therefore, to maximize utility we solve the following maximization:
\begin{align}
    \max\limits_{V_1 \subset V} \ \ \frac{N(V_1)+|V-V_1|}{n}
\end{align}
The following claim proves the theorem.
\begin{claim}\label{claim:dominating_set}
$G$ has a dominating set of size less than or equal to $k$ if and only if the optimal welfare is at least $\welf^* \ge 2-\frac{k}{n}$.  
\end{claim}
\begin{proof}
Suppose $G$ has a dominating set of size $\leq k$, then $\welf^* \ge 2-\frac{k}{n}$. To see that, set $V_1$ equal to the dominating set. Therefore, it follows that $N(V_1)=n$ and $|V-V_1|\ge n-k$ so by plugging these values in \eqref{eq:domset_plug} we get a welfare of at least $2-\frac{k}{n}$. 

For the other direction. We first show that if a solution has a welfare $\welf \ge 2-\frac{k}{n}$ then $|V_1|\leq k$. This is the case since if $|V_1|>k$ then $|V-V_1|< n-k$ and therefore using \eqref{eq:domset_plug} we get $\welf < 2-\frac{k}{n}$.

Now we argue that if $|V_1|=k-\ell \leq k$  where $\ell \ge 0$ is a non-negative integer then $N(V_1) \ge n-\ell$. Suppose not, i.e.,  $N(V_1) \leq n-\ell-1$ then we have 
\begin{align*}
    \welf &= \frac{N(V_1) + |V-V_1|}{n} \\
          & \leq \frac{n-\ell-1 + n - (k-\ell)}{n} \\
          & = \frac{2n - k -1}{n} \\
          & < 2-\frac{k}{n}
\end{align*}
Therefore, it must be that if $|V_1|=k-\ell$ then $N(V_1) \ge n-\ell$, but this implies that there exists a dominating set of size $k$. To see that, note that we can construct a dominating set from $V_1$ by adding (at most) the remaining $\ell$ many uncovered vertices, this leads to $N(V_1)+\ell = n-\ell+\ell=n$ using only $k-\ell+\ell=k$ many vertices which is a dominating set of size at most $k$. 
\end{proof}

\subsection{Proof of Second Claim in Theorem \ref{thm:welfare-hard}}\label{app:proof_welfare_NP_hard_m3}
\begin{proof}
The reduction is from the domatic number problem. First, given a graph $G=(V,E)$ a domatic partition of the graph is a set of disjoint vertices $V_1,\dots,V_D$ such that for each $i \in \{1,\dots,D\}$ the subset of vertices $V_i$ is a dominating set of $G$. In the decision version of the domatic number problem we are given a graph $G=(V,E)$ and an integer $k$ and we should decide if there is a domatic partition of $G$ into at least $k$ many subsets. In \citet{kaplan1994domatic} it is shown that the domatic number problem is NP-complete for $k \ge 3$.

We will now reduce any instance of the domatic number problem for $k \ge 3$ to our welfare maximization problem. First, note for a given value of $k$ if the graph has a domatic partition for $k' > k$ then it has a domatic partition of size $k$ this follows since given a domatic partition $\{V_1,\dots,V_{k'}\}$ we can reduce the number of subsets (possibly recursively) until we have a total of $k$ many subsets by simply merging two subsets since the union of two dominating sets is also a dominating set.  

The reduction is as follows. We will have $k$ bidders and for simplicity, we will set $b_1=b_2=\dots=b_k=b>0$. Given the graph $G=(V,E)$ the vertices $V$ will represent the database $\db$. All weights $w_{ij}=\frac{1}{n} \ , \forall j \in \db, \forall i \in [k]$. The distance between two points (vertices) will equal the path distance, hence vertices that have an edge between them are at a distance of $1$. For our coverage function, we set the radius $r_{ij}=r=1 \ , \forall j \in \db, \forall i \in [k]$. Therefore, if a bidder receives a subset of points $V' \subset V$ then his utility would be $v_i(V')=b \frac{N(V')}{n}$ where $N(V')$ is the total number of vertices that are either in $V'$ or adjacent to a vertex in $V'$. We will now show that there is a domatic partition into $k$ subsets if and only if the welfare is at least $b\cdot k$ in our problem. 

First, suppose that there exists a domatic partition into $k$ subsets $V_1, V_2, \dots, V_k$ then it follows since each subset is dominating that the welfare is at least $\welf = v_1(V_1)+v_2(V_2)+\dots+v_k(V_k)= b \cdot (\frac{n}{n}+\frac{n}{n}+\dots+\frac{n}{n}) = b \cdot k$. 

For the other direction, suppose that we have found an allocation with $\welf \ge bk$. By definition of the valuation function, it follows that for any bidder and any allocation $v_i(V') = b \frac{N(V')}{n} \leq b \frac{n}{n}=b$. But since $w \ge b \cdot k$ it follows that each bidder covers the entire set of points (vertices). Further, since we allocate each point (vertex) to only one bidder it follows that we have a partition of the graph into $k$ disjoint subset of vertices each covering the entire set of points, i.e., a partition into $k$ dominating sets.   
\end{proof}

\subsection{Proof of Lemma \ref{lem:opt_lp}}\label{app:proof_opt_lp}
\begin{proof}
To prove the lemma, we will show that an optimal solution is feasible and that objective \eqref{lp:obj} gives its welfare value. First, note that an optimal solution which consists of a collection of datasets $\{\bundle^*_1,\dots,\bundle^*_m\}$ corresponds to a set of integral assignments $X_{ij}^*$ which for a point $j\in \db$ and bidder $i \in [m]$ would equal 1 only if bidder $i$ received point $j$ and would be zero otherwise. 

Since $X_{ij}^*$ is an integral assignment ($X_{ij}^* \in \{0,1\}$) it follows that $0 \leq X_{ij}^*\leq 1$ satisfying constraint \eqref{lp_const:between_0and1}. Further, since no point can be assigned to more than one bidder we have $\sum_{i \in [m]} X_{ij}^* \leq 1$ satisfying constraint \eqref{lp_const:to_atmost_onebidder}. Moreover, the values $X_{ij}^*$ would induce a collection of coverage values $C_{ij}^*$. Since $C_{ij}^* \in \{0,1\}$ then we have $0 \leq C_{ij}^*\leq 1$ satisfying \eqref{lp_const:between_0and1}. Finally, a point $j$ is covered if and only if at least one point $j' \in N_i(j)$ is given to bidder $i$. If this is the case for a point $j$ then $\sum_{j' \in N_i(j)} X^*_{ij'}  \ge 1$ and $C^*_{ij}=1$. If no point in the neighborhood of $j$ is allocated to bidder $i$ then we must have $\sum_{j' \in N_i(j)} X^*_{ij'}=0$ and $C^*_{ij}=0$. In both cases, we have $C^*_{ij} \leq \sum_{j' \in N_i(j)} X^*_{ij'}$ satisfying constraint \eqref{lp_const:between_0and1}. From the above, it follows that $\sum_{i \in [m]} \theta_i \cdot \Big( \sum_{j \in  \db} w_{ij} \cdot C^*_{ij} \Big) = \welf^*(\bidprof)$. 

Since the LP also enables fractional assignments in $[0,1]$ not just integral assignments in $\{0,1\}$ for $\bX_{ij}$ and $\bC_{ij}$ it follows that $\optlp(\bidprof) \ge \welf^*(\bidprof)$.  
\end{proof}

\subsection{Proof of Lemma \ref{lem:rounding_approx}}\label{app:proof_rounding_approx}
\begin{proof}
\begin{align}
    \E(C_{ij}=1)  &= 1-\Prob(C_{ij}=0) \\ 
                    & = 1- \Prob(\forall j' \in N_i(j), X_{ij'}=0) \\ 
                    & = 1-\prod_{j' \in N_i(j)} \Prob( X_{ij'}=0) \quad \quad \text{(Since the rounding is independent)} \nonumber \\
                    & = 1 - \prod_{j' \in N_i(j)} (1-\bX_{ij'}) \label{eq:proved_1} \\ 
                    & \ge 1 - \prod_{j' \in N_i(j)} \exp(-\bX_{ij'})  \quad \quad \text{(Since $1-x\leq e^{-x}$)} \nonumber \\
                    & = 1 - \exp(-\sum_{j' \in N_i(j)} \bX_{ij'}) \\ 
                    & \ge 1 - \exp( -\bC_{ij}) \quad \quad \text{(Since $\bC_{ij} \leq \sum_{j' \in N_i(j)} \bX_{ij'}$ by constraint \eqref{lp_const:between_0and1})} \nonumber\\
                    & \ge (1-1/e) \bC_{ij} \label{eq:proved_2} \quad \quad \text{(By Fact \ref{fact:lower_bound_on_exp} and \eqref{lp_const:between_0and1} we have $0 \leq \bC_{ij}\leq 1$)} \nonumber
\end{align}
Line \eqref{eq:proved_1} proves Equation \eqref{eq:point_cov_exact_form}. The second point follows immediately by definition \eqref{eq:cov_i_def} and Line \eqref{eq:proved_2}. 
\end{proof}

\subsection{Proof of Theorem \ref{th:lp_round_approx}}
\begin{proof}
The resulting welfare is 
\begin{align}
    \E\Big[\sum_{i \in [m]} \theta_i \cdot  \cov_i \Big] & = \sum_{i \in [m]} \theta_i \cdot  \E[\cov_i]  \\ 
    & \ge (1-1/e) \sum_{i \in [m]} \theta_i \cdot \bcov_i \\
    & =   (1-1/e) \optlp(\bidprof)  \\ 
    & \ge (1-1/e) \welf^*(\bidprof)
\end{align}
\end{proof}

\section{Omitted Proofs of Section \ref{sec:new_ic_theory}}
\subsection{Proof of Lemma \ref{th:gen_myerson_lemma}}\label{app:proof_gen_myerson_lemma}
\begin{proof}
We will focus on a single bidder $i \in  [m]$ and hold the other buyer's bids $\bidprof_{-i}$ fixed. Bidder $i$ will receive datasets $\bundle$ and  $\bundle'$ under bid profiles $(b,\bidprof_{-i})$ and $(b',\bidprof_{-i})$, respectively. Note that we have dropped the subscript $i$ in  $\bundle$,$\bundle'$,$b$, and $b'$ for ease of notation. We also overload the notation  and have $p(b)=p(\bundle)$ and $p(b')=p(\bundle')$.

For the mechanism to be incentive compatible the following conditions must hold
\begin{align}
 b \cdot \E[\cov_i(\bundle)] - \E[p(b)] & \ge b \cdot \E[\cov_i(\bundle')] - \E[p(b')] \label{eq:ml_1}\\ 
 b' \cdot \E[\cov_i(\bundle')] - \E[p(b')] & \ge b' \cdot \E[\cov_i(\bundle)] - \E[p(b)] \label{eq:ml_2}
\end{align}
From \eqref{eq:ml_1} and \eqref{eq:ml_2} the following holds
\begin{align}
    (b-b') \cdot (\E[\cov_i(\bundle)] -\E[\cov_i(\bundle')]) \ge 0 \label{eq:ml_3}
\end{align}
The above inequality immediately implies that the allocation must be monotonic. I.e., $\E[\cov_i(\bundle)] \ge \E[\cov_i(\bundle')$ if and only if $b \ge b'$. This proves the first part of the theorem.  

To find the payment, first note that \eqref{eq:ml_1} and \eqref{eq:ml_2} imply the following inequalities:
\begin{align}
     b' \cdot \Big( \E[\cov_i(\bundle)] - \E[\cov_i(\bundle')] \Big)  
     \leq \E[p(b)] - \E[p(b')] \label{eq:ml_p1}\\ 
     \leq b \cdot \Big( \E[\cov_i(\bundle)] - \E[\cov_i(\bundle')] \Big) \label{eq:ml_p2}
\end{align}
We will now set the payment rule. Note that we will follow a deterministic payment rule, thus the payment value will always equal its expectation. The possible bid values belong to the set $\bidset=\{\lambda_0,\lambda_2,\dots,\lambda_{|\bidset|}\}$. 
%
We denote by $\bundle_{t}$ the dataset which bidder $i$ receives if he bids with value $\lambda_t$ for some $t \in \{0,1,2,\dots,|\bidset|\}$ with the other bidders $\bidprof_{-i}$ held fixed. Further, we set $\bundle_{0}=\emptyset$ hence $\E[\cov_i(\bundle_{0})]=0$ and we also set $p(0)=0$. At a bid value of $b$ the payment for bidder $i$ is the following:
\begin{align} \label{eq:payment_rule}
    p(b) = \sum_{\substack{t \in  \{0,1,2,\dots,|\bidset|\}: \\ \lambda_t < b}} \lambda_t \cdot (\E[\cov_i(\bundle_{t+1})]-\E[\cov_i(\bundle_{t})])
\end{align}
Without loss of generality, we assume that for the two bid values $b$ and $b'$ we have $b>b'$. We will now prove that Inequalities \eqref{eq:ml_p1} and \eqref{eq:ml_p2} are satisfied by this payment. 
\begin{align}
& b' \cdot \Big( \E[\cov_i(\bundle)] - \E[\cov_i(\bundle')] \Big)  & \\
& = b'  \Big( \sum_{\substack{t \in  \{0,1,2,\dots,|\bidset|\}: \\ b'\leq \lambda_t < b}}  (\E[\cov_i(\bundle_{t+1})]-\E[\cov_i(\bundle_{t})]) \Big) 
&\qquad \text{(since it is a telescoping sum)} \nonumber \\
& \leq \sum_{\substack{t \in  \{0,1,2,\dots,|\bidset|\}: \\ b'\leq \lambda_t < b}}  \lambda_t \cdot (\E[\cov_i(\bundle_{t+1})]-\E[\cov_i(\bundle_{t})]) \label{eq:s1}\\
& = p(b) - p(b') \quad \text{(by Equation \eqref{eq:payment_rule})} \label{eq:s2}
\end{align}
This proves inequality \eqref{eq:ml_p1}. Further, from Lines \eqref{eq:s1} and \eqref{eq:s2} we have
\begin{align*}
    p(b) - p(b') & = \sum_{\substack{t \in  \{0,1,2,\dots,|\bidset|\}: \\ b'\leq \lambda_t < b}}  \lambda_t \cdot (\E[\cov_i(\bundle_{t+1})]-\E[\cov_i(\bundle_{t})]) \\ 
    & \leq b \sum_{\substack{t \in  \{0,1,2,\dots,|\bidset|\}: \\ b'\leq \lambda_t < b}}   (\E[\cov_i(\bundle_{t+1})]-\E[\cov_i(\bundle_{t})]) \\
    & = b \cdot \Big( \E[\cov_i(\bundle)] - \E[\cov_i(\bundle')] \Big)
\end{align*}
This proves Inequality \eqref{eq:ml_p2}. This shows a payment rule that can be efficiently computed and that would lead to an incentive compatible mechanism. Note that we assume that the expected value of the allocation can be found for any bid value. Our main algorithm (Algorithm \ref{alg:lp_and_round_monotonic}) in Section \ref{sec:new_ic_theory} gives such an allocation. 

We now prove that individual rationality holds. Let $D_{\theta_i}$ be the dataset allocated when $i$ bids with value $\theta_i$ then we have 
\begin{align*}
    \theta_i \E[\cov_i(D_{\theta_i})] - p(\theta_i)  & \ge \theta_i \E[\cov_i(D_{1})] - p(\lambda_1)  \quad \quad \text{(since the mechanism is incentive compatible)} \nonumber \\
    & = \theta_i \E[\cov_i(D_{1})] - \lambda_0 \E[\cov_i(D_{1})] \quad \quad \text{(using \eqref{eq:payment_rule})} \\ 
    & =  (\theta_i - \lambda_0) \E[\cov_i(D_{1})] \\
    & \ge 0 
\end{align*}
\end{proof}

\subsection{Proof of Fact \ref{thm:lp_round_not_monotone}}\label{app:proof_lp_round_not_monotone}
\begin{proof}
This can be shown by an example. Specifically, suppose that we only have three points and two bidders ($n=3$ and $m=2$). Let the radius values be the same for all bidders and points, i.e., $r_{ij}=r, \ \forall j \in \db, \forall i \in [m]$. Further, suppose that all of the points are in each others neighborhoods, i.e., $\forall j,j' \in \db: d(j,j') \leq \rad$. Moreover, we set $w_{ij}=\frac{1}{3} \ \forall j \in \db, \forall i \in [m]$.

Suppose the second bidder always has $b_2=1$ but the first bidder has $b_1=2 $ and $b'_1 =1$. Consider the following LP optimal solution under $b'_1$:
\[
\begin{array}{c c c}
\bX'_{11} = 1 & \bX'_{12} = 0 & \bX'_{13} = 0 \\ 
\bX'_{21} = 0 & \bX'_{22} = 1 & \bX'_{23} = 1 \\ 
\end{array}
\]
Note that this leads to an optimal coverage of $1$ for both bidders in the LP solution and when rounded, i.e., $\bcov_1=\bcov_2=\E[\cov_1]=\E[\cov_2]=1$. 

Now consider the following LP optimal solution under $b_1$: 
\[
\begin{array}{c c c}
\bX_{11} = \frac{1}{2} & \bX'_{12} = \frac{1}{2} & \bX_{13} = \frac{1}{2} \\ 
\bX_{21} = \frac{1}{2} & \bX'_{22} = \frac{1}{2} & \bX_{23} = \frac{1}{2} \\ 
\end{array}
\]
Although the LP solution is optimal for both bidders, the rounded allocation is smaller in expectation. This follows since by Equation \eqref{eq:point_cov_exact_form} we have $\E[\cov_1] = \E[\cov_2] = 1-\frac{1}{2^3} < 1$.    
\end{proof}

\subsection{Proof of Theorem \ref{thm:lp_round_monotone}}\label{app:proof_lp_round_monotone}
\begin{proof}
First we start with the following claim that shows that the LP coverage values are monotonic.
\begin{proof}[Proof of Lemma \ref{lemma:lp_is_monotonic}]
This follows immediately as a special case of Lemma \ref{lemma:gen_monotonicity}. 
\end{proof}
Now we establish the validity of the probability values $\pval_i$ for each bidder $i \in [m]$.
\begin{claim}
For each bidder $i \in [m]$ we have $\pval_i \in [0,1]$.
\end{claim}
\begin{proof}
First, if $\bcov_i=0$ then it follows that $\pval_i=0$ and is therefore in $[0,1]$. So we focus now on the case where $\bcov_i>0$, we have the following:
\begin{align}
    \pval_i & = (1-1/e) \cdot \frac{\bcov_i}{\E[\cov_i]} \\
    & \leq (1-1/e) \cdot \frac{1}{(1-1/e)}  \quad \quad \text{(By Inequality \eqref{eq:point_cov_exact_form})}  \\
    & = 1 
\end{align}
Further, it is clear that we would have $\pval_i >0$. 
\end{proof}
We now introduce the following lemma which lower bounds the final expected coverage $\E[\cov'_i]$ a bidder $i$ receives from Algorithm \ref{alg:lp_and_round_monotonic}.
\begin{proof}[Proof of Lemma \ref{lem:expected_ceverage}]
Note that when a non-empty dataset is allocated to the bidder the expected coverage will be the same as the one from Algorithm \ref{alg:lp_and_round} which is $\E[\cov_i]$. Therefore, we have 
\begin{align}
    \E[\cov'_i] & = (1-\pval_i) \cdot 0 + \pval_i \cdot \E[\cov_i] \\ 
                & = (1-1/e) \cdot \frac{\bcov_i}{\E[\cov_i]} \cdot \E[\cov_i] \\
                & = (1-1/e) \cdot \bcov_i
\end{align}
\end{proof}
Now we can prove the main claims of the theorem. We start by proving the approximation ratio. By the above lemma we immediately have $\sum_{i \in [m]} \theta_i \E[\cov'_i] = (1-1/e) \cdot \sum_{i \in [m]} \theta_i   \bcov_i = (1-1/e) \cdot \optlp(\bidprof) \ge (1-1/e) \cdot \welf^*(\bidprof)$. 

Now to prove monotonicity, note that by the Lemma each bidder receives exactly coverage of $(1-1/e) \cdot \bcov_i$. Therefore, if a bidder deviates from $b_i$ to $b_i'$ then we have $(1-1/e) \cdot \bcov_i \ge (1-1/e) \cdot \bcov_i'$ if $b_i \ge b'_i$ by Lemma \ref{lemma:lp_is_monotonic}.
\end{proof}

\section{Revenue Maximizing Auction using Virtual Welfare}\label{app:revenue}
We can convert our welfare maximizing results to a $(1-1/e)$-approximate revenue maximizing auction using the standard Myerson virtual valuation approach.
The revenue maximizing mechanism requires knowledge of the distribution of buyer types, and in this setting, this corresponds to a knowledge of the distribution of $\theta$.
We will assume that each buyer's valuation $\theta_i$ is independent of all other buyers, and that $\theta_i \sim F_i(\theta)$ where $F_i(\theta_i)$ is the cumulative distribution function, denote by $f_i(\theta_i)$ the probability density function.
Note that throughout this section, we differ from the notation in the main body of the paper in that we refer to the set of bidder types as $\Theta_i$ for each type where $\boldsymbol{\Theta} = \times_i \Theta_i$ and $\boldsymbol{\Theta}_{-i} = \times_{j\ne i} \Theta_j$.
Additionally, we refer to the ordered discrete set of types for bidder $i$ as $\{\theta_{i,1}, \theta_{i,2}, ..., \theta_{i, |\Theta|}\}$.
Finally, we denote the coverage that bidder $i$ received given bidder profile $(\theta_i, \boldsymbol{\theta}_{-i})$ as $\cov_i(\theta_{i}, \boldsymbol{\theta}_{-i}))])$.
We can then write the expected revenue from a truthful mechanism as follows:
\begin{equation}
    \E_{F({\boldsymbol{\theta}})}\left[\sum_{i \in \{m\}} p_i(\theta_i, \boldsymbol{\theta}_{-i})\right] = \sum_{i \in \{m\}} \E_{F({\boldsymbol{\theta}})}[p_i(\theta_i, \boldsymbol{\theta}_{-i})]
\end{equation}
By equation~\ref{eq:payment_rule}, 
\begin{align*}
    p(\theta_{i,t}, \boldsymbol{\theta}_{-i}) = \sum_{\substack{t \in  \{0,1,2,\dots,|\Theta_i|\}: \\ \theta_{i,t} < \theta_i}} \theta_{i,t} \cdot (\E[\cov_i(\theta_{i, t+1}, \boldsymbol{\theta}_{-i})]-\E[\cov_i(\theta_{i, t}, \boldsymbol{\theta}_{-i})])
\end{align*}
Therefore,
\begin{small}
\begin{align*}
    &\E_{F({\boldsymbol{\theta}})}[p_i(\theta_i, \boldsymbol{\theta}_{-i})] = \E_{F({\boldsymbol{\theta}})}\left[\sum_{\substack{t \in  \{0,1,2,\dots,|\Theta_i|\}: \\ \theta_{i,t} < \theta_i}} \theta_{i,t} \cdot (\E[\cov_i(\theta_{t+1}, \boldsymbol{\theta}_{-i})]-\E[\cov_i(\theta_{i,t}, \boldsymbol{\theta}_{-i}))])\right] \\
    &= \sum_{\boldsymbol{\theta} \in \boldsymbol{\Theta}} f_{-i}(\boldsymbol{\theta})\sum_{\substack{t \in  \{0,1,2,\dots,|\Theta_i|\}: \\ \theta_{i,t} < \theta_i}} \theta_{i,t} \cdot (\E[\cov_i(\theta_{t+1}, \boldsymbol{\theta}_{-i}))]-\E[\cov_i(\theta_{i,t}, \boldsymbol{\theta}_{-i})]) \\
    &= \sum_{\boldsymbol{\theta_{-i}} \in \boldsymbol{\Theta_{-i}}} f_{-i}(\boldsymbol{\theta}_{-i})\left( \sum_{\theta_i \in \Theta_i} f(\theta_i)\sum_{\substack{t \in  \{0,1,2,\dots,|\Theta_i|\}: \\ \theta_{i,t} < \theta_i}} \theta_{i,t} \cdot (\E[\cov_i(\theta_{t+1}, \boldsymbol{\theta}_{-i}))]-\E[\cov_i(\theta_{i,t}, \boldsymbol{\theta}_{-i})])\right) \\
    &= \sum_{\boldsymbol{\theta_{-i}} \in \boldsymbol{\Theta_{-i}}} f_{-i}(\boldsymbol{\theta}_{-i}) \sum_{t \in  \{0,1,2,\dots,|\Theta_i|\}} \left(\sum_{\substack{\theta_i \in \Theta_i: \\ \theta_i \ge \theta_{i,t}}} f(\theta_i) \theta_{i,t} \cdot (\E[\cov_i(\theta_{t+1}, \boldsymbol{\theta}_{-i}))]-\E[\cov_i(\theta_{i,t}, \boldsymbol{\theta}_{-i})])
    \right) \\
    &= \sum_{\boldsymbol{\theta_{-i}} \in \boldsymbol{\Theta_{-i}}} f_{-i}(\boldsymbol{\theta}_{-i}) \sum_{t \in  \{0,1,2,\dots,|\Theta_i|-1\}} \left((1 - F(\theta_{i,t})) \theta_{i,t} \cdot (\E[\cov_i(\theta_{t+1}, \boldsymbol{\theta}_{-i}))]-\E[\cov_i(\theta_{i,t}, \boldsymbol{\theta}_{-i})])
    \right) \\
    &= \sum_{\boldsymbol{\theta_{-i}} \in \boldsymbol{\Theta_{-i}}} f_{-i}(\boldsymbol{\theta}_{-i}) \left( -(1-F_i(\theta_{i,0}))\cdot \theta_{i,0} \cdot \E[\cov_i(\theta_{i,0}, \boldsymbol{\theta}_{-i})]\right) \\
    & \qquad - \sum_{t \in  \{0,1,2,\dots,|\Theta_i|-1\}}\E[\cov_i(\theta_{t+1}, \boldsymbol{\theta}_{-i})]((1 - F(\theta_{t+1})) \theta_{t+1} - (1 - F(\theta_{i,t})) \theta_{i,t})) \\
    &= \sum_{\boldsymbol{\theta_{-i}} \in \boldsymbol{\Theta_{-i}}} f_{-i}(\boldsymbol{\theta}_{-i}) \left( -(1-F_i(\theta_{i,0}))\cdot \theta_{i,0} \cdot \E[\cov_i(\theta_{i,0}, \boldsymbol{\theta}_{-i})]\right) \\
    & \qquad + \sum_{t \in  \{0,1,2,\dots,|\Theta_i|-1\}}f(\theta_{i,t})\E[\cov_i(\theta_{t+1}, \boldsymbol{\theta}_{-i})](\theta_{t+1} - (\theta_{t+1} - \theta_{i,t})\frac{1 - F(\theta_{i,t})}{f(\theta_{i,t})})) \\
\end{align*}
\end{small}
Note that the first term $-(1-F_i(\theta_{i,0}))\cdot \theta_{i,0} \cdot \E[\cov_i(\theta_{i,0}, \boldsymbol{\theta}_{-i})]$ is always negative, so we can maximize revenue by ensuring that we do not allocate to any type with the lowest possible $\theta_i$.
This will ensure that the first term is zero, and it leaves the revenue maximizing equation as:
\begin{small}
\begin{align*}
    &\sum_{i \in \{m\}} \E_{F({\boldsymbol{\theta}})}[p_i(\theta_i, \boldsymbol{\theta}_{-i})] = \\
    &\sum_{i \in \{m\}} \sum_{\boldsymbol{\theta_{-i}} \in \boldsymbol{\Theta_{-i}}} f_{-i}(\boldsymbol{\theta}_{-i})\sum_{t \in  \{0,1,2,\dots,|\Theta_i|\}}f(\theta_{i, t})\E[\cov_i(\theta_{i, t+1}, \boldsymbol{\theta}_{-i})](\theta_{i, t+1} - (\theta_{i, t+1} - \theta_{i, t})\frac{1 - F(\theta_{i, t})}{f(\theta_{i, t})})) \\
    &= \E_{\boldsymbol{\theta}} [\sum_{i \in \{m\}} (\theta_{i, t+1} - (\theta_{i, t+1} - \theta_{i, t})\frac{1 - F(\theta_{i, t})}{f(\theta_{i, t})}) \E[\cov_i(\theta_{i, t+1}, \boldsymbol{\theta}_{-i})]]
\end{align*}
\end{small}
Note that we can treat the first term, $(\theta_{i, t+1} - (\theta_{i, t+1} - \theta_{i, t})\frac{1 - F(\theta_{i, t})}{f(\theta_{i, t})})$, as a scaled type, a virtual type, denoted $\phi(\theta_{i,t+1})$.
Then if we maximize $\E_{\boldsymbol{\theta}}[\phi(\theta_{i,t+1})\E[\cov_i(\theta_{i, t+1}, \boldsymbol{\theta}_{-i})]]$ we will maximize revenue.
However, this is just a welfare maximization problem with virtual types replacing the true type.
Therefore, a monotone allocation algorithm that maximizes welfare for the virtual types, will maximize revenue.
Also, if an allocation algorithm approximately maximizes virtual welfare, then it will also approximately maximize revenue with the same approximation factor.
Note that, as is standard in revenue maximization through virtual values, if the virtual value is not monotone in the true type, we can use an ironing procedure to ensure monotonicity.
Then the monotonicity of the virtual value in the true type ensures incentive compatibility and individual rationality as in Lemma~\ref{th:gen_myerson_lemma}.

\section{$(1-1/e)$ Welfare Maximizing Algorithm for Standard Coverage Functions}\label{app:alg_standard_coverage}
First, we start by introducing the coverage functions as defined in \citet{dobzinski2006improved}. Note that for the setting,  \citet{dobzinski2006improved} algorithm requires using the demand oracle whereas \citet{dughmi2011convex} runs in expected polynomial time. We have a ground set $U$ and a collection of $L$ many subsets of $U$ $R_1,\dots,R_L$. We have $m$ many bidders (buyers). We can allocate each set in $\{R_1,\dots,R_L\}$ to only one bidder. If a buyer $i \in [m]$ receives a collection of sets $R_{i,1},\dots, R_{i,L_i}$ then his valuation would be $v_i(\{R_{i,1},\dots, R_{i,L_i}\}) = \theta_i \cdot |\cup_{\ell \in [L_i]} R_{i,\ell}|$  where $\theta_i \ge 0$ is the private type (importance) parameter. Note this generalizes the function from \citet{dobzinski2006improved}  since we associate a parameter $\theta_i$ for each bidder $i$ whereas in \citet{dobzinski2006improved} all bidders have $\theta_i=1$.

For an element $j \in U$ we define $N(j)=\{R_\ell| j \in R_\ell\}$. Now we introduce the below LP which is essentially modified from \ref{opt:lp}. 
\begin{subequations}  \label{cov_opt:ilp}  
\begin{equation} \label{cov_ilp:obj}  
\max\limits_{\bX_{i\ell},\bC_{ij}}  \quad \sum_{i \in [m]} \theta_i \Big( \sum_{j \in  U} \bC_{ij} \Big) 
\end{equation}    
\begin{equation}   \label{cov_ilp_const:integ}
\forall i\in [m], \forall j \in U, \forall \ell \in [L]: \quad 0 \leq  \bX_{i\ell} , \bC_{ij} \leq 1 
\end{equation}
\begin{equation}   \label{cov_ilp_const:to_atmost_onebidder}
\forall \ell \in [L]: \quad  \sum_{i \in [m]} \bX_{i\ell} \leq 1 
\end{equation}
\begin{equation}   \label{cov_ilp_const:coverage_at_most1}
\forall i\in [m], \forall j \in U: \quad \bC_{ij} \leq \sum_{\ell \in N(j)} \bX_{i\ell} 
\end{equation}
\end{subequations}
The same methods of Section \ref{sec:new_algorithmic} can be followed straightforwardly to show a $1-1/e$ approximation. Further, a similar data burning method from \ref{sec:new_ic_theory} can be used to guarantee monotonicity.

\section{Counterexample Demonstrating the Non-Monotonicity of the $\cg$ Algorithm}\label{app_sec:counterexample}

As introduced in Section \ref{sec:experiment}, our $\syn$ dataset has $m=2$, $n=5$, $|\Lambda|=10$, $d=10$. And the dataset $\db$ is generated uniformly within the range of $[0, 10]$ with a shape of $n \times d$. Specifically,  the randomly generated buyer types are $\theta_1 = 0.4$ and $\theta_2=0.7$. The  weights $w_{ij} = 0.2$ for all $i$ and $j$. The radius $r_{ij} = 11.50893871$ for all $j$ when $i=1$ is the first buyer; and $r_{ij} = 10.48592194$ for all $j$ when $i=2$. Finally, the randomly generated dataset $\db$ is as follows.

\[
\db =
\resizebox{0.95\textwidth}{!}{$
\begin{bmatrix}
2.60299691 & 8.70395688 & 1.85039927 & 0.19661425 & 9.53252032 & 6.80450805 & 4.86588127 & 9.6502682  & 3.93398739 & 0.79557571 \\
3.51407424 & 1.63635163 & 9.83166821 & 8.80628184 & 4.94063468 & 4.00959241 & 4.51291463 & 7.20876849 & 2.47768284 & 6.22779952 \\
1.42448816 & 2.01176282 & 0.81217729 & 9.53472295 & 0.5573827  & 5.99536483 & 7.2299763  & 9.70289719 & 8.21569457 & 5.27551067 \\
3.3147673  & 3.53982196 & 0.79030301 & 5.5591438  & 1.65794463 & 2.95237831 & 8.40608485 & 3.62562009 & 3.40813978 & 7.10401744 \\
0.31707752 & 4.72387866 & 7.31576623 & 7.73215774 & 2.41244792 & 3.06197046 & 1.36649638 & 5.9573572  & 1.62568866 & 8.02327214
\end{bmatrix}
$}
\]




\end{document}